%% file: main.tex
\documentclass[a4paper,UKenglish,cleveref, autoref]{lipics-v2019}

\input{0_macros}

\title{On Succinctness and Recognisability of Alternating Good-for-Games Automata} 

\author{Udi Boker}{Interdisciplinary Center (IDC) Herzliya, Israel}{udiboker@gmail.com}{}{Israel Science Foundation grant 1373/16}

\author{Denis Kuperberg}{CNRS, LIP, \'Ecole Normale Supérieure, Lyon, France}{denis.kuperberg@ens-lyon.fr}{0000-0001-5406-717X}{}

\author{Karoliina Lehtinen}{University of Liverpool, United Kingdom}{k.lehtinen@liverpool.ac.uk}{0000-0003-1171-8790}{EPSRC grant EP/P020909/1  (Solving Parity Games in Theory and Practice)}

\author{Micha\l~Skrzypczak}{Institute of Informatics, University of Warsaw, Poland}{mskrzypczak@mimuw.edu.pl}{0000-0002-9647-4993}{}

\authorrunning{U. Boker, D. Kuperberg, K. Lehtinen, M. Skrzypczak}

\Copyright{Udi Boker, Denis Kuperberg, Karoliina Lehtinen, Micha\l~Skrzypczak}

\ccsdesc[300]{Theory of computation~Logic and verification}

\keywords{Good for games, history-determinism, alternation}

\category{}

\supplement{}

\acknowledgements{}

\nolinenumbers 

\hideLIPIcs  

\EventEditors{John Q. Open and Joan R. Access}
\EventNoEds{2}
\EventLongTitle{42nd Conference on Very Important Topics (CVIT 2016)}
\EventShortTitle{CVIT 2016}
\EventAcronym{CVIT}
\EventYear{2016}
\EventDate{December 24--27, 2016}
\EventLocation{Little Whinging, United Kingdom}
\EventLogo{}
\SeriesVolume{42}
\ArticleNo{23}

\begin{document}
\maketitle

\begin{abstract}
We study alternating good-for-games (GFG) automata, i.e., alternating automata where both conjunctive and disjunctive choices can be resolved in an online manner, without knowledge of the suffix of the input word still to be read. We show that they can be exponentially more succinct than both their nondeterministic and universal counterparts.
Furthermore, we lift many results from nondeterministic parity GFG automata to alternating ones: a single exponential determinisation procedure, an \exptime upper bound to the GFGness problem, a \ptime algorithm for the GFGness problem of weak automata, 
and a reduction from a positive solution to the $G_2$ conjecture to a \ptime algorithm for the GFGness problem of parity automata with a fixed index. The $G_2$ conjecture states that a nondeterministic parity automaton $\A$ is GFG if and only if a token game, known as the $G_2$ game, played on $\A$ is won by the first player. So far, it had only been proved for B\"uchi automata; we provide further evidence for it by proving it for coB\"uchi automata.
We also study the complexity of deciding ``half-GFGness'', a property specific to alternating automata that only requires nondeterministic choices to be resolved in an online manner. We show that this problem is strictly more difficult than GFGness check, already for alternating automata on finite words.
\end{abstract}

\newpage

\input{1_intro}

\input{2_prelim}

\input{3_alt-gfg}

\input{4_determ}

\input{5_deciding}

\input{6_cobuchi}

\input{7_concl}

\newpage

\bibliography{gfg}
\vspace{2cm}

\noindent{\huge{\textbf{Appendix}}}
\appendix

\input{AP_2_prelim}

\input{AP_3_alt-gfg}

\input{AP_4_determ}

\input{AP_5_deciding}

\input{AP_6_cobuchi}

\end{document}

%% file: 0_macros.tex
\usepackage{amsmath}
\usepackage{amsfonts}
\usepackage{amssymb}
\usepackage{xspace}
\usepackage[table,xcdraw]{xcolor}
\usepackage{graphicx}
\usepackage{multicol}

\usepackage{thm-restate}

\newtheorem{conjecture}[theorem]{Conjecture}

\usepackage[shortcuts]{extdash}

\usepackage{tikz}

\usetikzlibrary{positioning}
\usetikzlibrary{decorations.pathreplacing}
\usetikzlibrary{automata,positioning}
\usetikzlibrary{decorations.pathmorphing}
\usetikzlibrary{decorations.markings}
\usetikzlibrary{decorations}
\usetikzlibrary{arrows}
\usetikzlibrary{patterns}
\usetikzlibrary{calc}
\usetikzlibrary{shapes}

\tikzstyle{ubrace} = [draw, thick, decoration={brace, amplitude=7pt, mirror, raise=0.0cm}, decorate,
    every node/.style={anchor=north, yshift=-0.15cm}]
\tikzstyle{rbrace} = [draw, thick, decoration={brace, amplitude=7pt, mirror, raise=0.0cm}, decorate,
    every node/.style={anchor=west, xshift= 0.15cm}]

\tikzstyle{obrace} = [draw, thick, decoration={brace, amplitude=7pt, raise=0.0cm}, decorate,
    every node/.style={anchor=south, yshift= 0.15cm}]
\tikzstyle{lbrace} = [draw, thick, decoration={brace, amplitude=7pt, raise=0.0cm}, decorate,
    every node/.style={anchor=east, xshift=-0.15cm}]

\tikzstyle{eve} = [draw, diamond, minimum size=20pt, inner sep=1pt]
\tikzstyle{adam} = [draw, regular polygon, regular polygon sides=4, inner sep=1pt]
\tikzstyle{state} = [circle, draw, inner sep=1pt, minimum size=20pt]
\tikzstyle{bool} = [circle, draw, inner sep=1pt, minimum size=5pt, scale=0.9]
\tikzstyle{dot} = [fill, circle, inner sep=0pt, minimum size=4pt]
\tikzstyle{dots} = [scale=2.0]
\tikzstyle{trans} = [draw, -{latex}, auto]
\tikzstyle{booltrans} = [draw, -{latex}, auto]
\tikzstyle{edge} = [draw, thick]
\tikzstyle{edgeBox} = [draw, thick,-latex]

\newcommand{\evalInt}[2]{%
	\pgfmathparse{int(#2)}%
	{\global\edef#1{\pgfmathresult}}%
}

\definecolor{darkgreen}{rgb}{0.13, 0.55, 0.13}
\newcommand{\todo}[1]{\textcolor{red}{\textsf{To do: #1}}}

\newcommand{\NotNeeded}[1]{}
\newcommand{\FullVersion}[1]{}

\newcommand{\Subject}[1]{\subparagraph{#1.}}
\newcommand{\ProofSketch}{\begin{proof}[Proof sketch.]}

\newcommand{\St}{~|~}

\newcommand{\tuple}[1]{\langle #1  \rangle}
\newcommand{\pair}{\tuple}

\newcommand{\Nat}{\ensuremath{\mathbb{N}}\xspace}

\newcommand{\Func}[1]{{\mathsf{#1}}}
\newcommand{\BP}{\Func{B}^+}

\newcommand{\pro}[2]{\ensuremath{#1 \times #2}\xspace} 

\newcommand{\strat}{\stratE}
\newcommand{\boxes}{\mathsf{B}}

\newcommand{\A}{{\cal A}}
\newcommand{\B}{{\cal B}}
\newcommand{\C}{{\cal C}}
\newcommand{\D}{{\cal D}}
\newcommand{\E}{{\cal E}}

\newcommand{\G}{{\cal G}}

\newcommand{\N}{{\cal N}}

\renewcommand{\S}{{\cal S}}

\newcommand{\restr}{{\upharpoonright}}

\newcommand{\iotar}{(\iota,\{\iota\})}

\newcommand{\Ssafe}{\mathsf{SafeZone}}
\newcommand{\delrr}{\Delta}

\newcommand{\Gsafe}{G_{\mathit{safe}}}
\newcommand{\Wsafe}{W_{\mathit{safe}}}
\newcommand{\trans}[3]{#1\xrightarrow[]{#2}#3}
\newcommand{\Ar}{\A_{\mathit{r}}}
\newcommand{\Qr}{Q_{\mathit{r}}}

\newcommand{\delr}{\delta_{\mathit{r}}}
\newcommand{\alpr}{\alpha_{\mathit{r}}}
\newcommand{\deltadet}{\delta_{\mathit{det}}}
\newcommand{\reach}[1]{\overline{#1}}
\newcommand{\sigmasafe}{\sigma_{\mathit{safe}}}
\newcommand{\sigmainf}{\sigma_{\infty}}
\newcommand{\Wk}{W_k}
\newcommand{\taukunif}{\tau_k^{\mathit{unif}}}
\newcommand{\sigmakunif}{\sigma_k^{\mathit{unif}}}
\newcommand{\sigmakunifM}{\sigma_{k,M}^{\mathit{unif}}}
\newcommand{\Gt}{G_2}

\newcommand{\EGFG}{\exists\mathrm{GFG}}
\newcommand{\AGFG}{\forall\mathrm{GFG}}

\newcommand{\Gpos}{G_{\mathtt{pos}}}
\newcommand{\GM}{{\cal P}}

\newcommand{\boxX}[1]{\mathtt{box}(#1)}
\newcommand{\boxA}{\boxX{\A}} 
\newcommand{\boxAco}{\boxX{\bar \A}}

\newcommand{\noLab}{\epsilon}

\newcommand{\stratE}{\sigma}
\newcommand{\stratA}{\tau}
\newcommand{\upd}{\mathit{upd}}
\newcommand{\letter}{a}

\newcommand{\pspace}{{\sc{PSpace}}\xspace}
\newcommand{\ptime}{{\sc{PTime}}\xspace}
\newcommand{\exptime}{{\sc{Exptime}}\xspace}

%% file: 1_intro.tex
\section{Introduction}
\label{sec:Introduction}

\emph{Good-for-games} (GFG) automata were first introduced in~\cite{HP06} as a tool for solving the synthesis problem.
The equivalent notion of \emph{history\=/determinism} was introduced independently in~\cite{Col09} in the context of regular cost functions.
Intuitively, a nondeterministic automaton is GFG if nondeterminism can be resolved on the fly, only with knowledge of the input word read so far.
GFG automata can be seen as an intermediate formalism between deterministic and nondeterministic ones, with advantages from both worlds. Indeed, like deterministic automata, GFG automata enjoy good compositional properties---useful for solving games and composing automata and trees---and easy inclusion checks~\cite{BKKS13}. Like nondeterministic automata, they can be exponentially more succinct than deterministic automata~\cite{KS15}.

In recent years, much effort has gone into understanding various properties of nondeterministic GFG automata, for instance their relationship with deterministic automata~\cite{BKKS13, KS15,BKS17,KM19}, applications in probabilistic model checking~\cite{KMBK14} and LTL and $\mu$\=/calculus synthesis~\cite{IK19}, decision procedures for GFGness~\cite{LR13,KS15,BK18}, minimisation~\cite{AK19}, and links with recent advances in parity games~\cite{CF19}.

Alternating GFG automata are a natural generalisation of nondeterministic GFG automata that enjoy the same compositional properties as nondeterministic GFG automata, while providing more flexibility. As we show in the present work, for some languages they can also be exponentially more succinct, allowing for better synthesis procedures.
Alternating GFG automata were introduced independently by Colcombet~\cite{colcombet2013hab} and Quirl~\cite{Quirl} while a~form of alternating GFG automata with requirements specific to counters were also considered in~\cite{KV11}, as a tool to study cost functions on infinite trees.
Boker and Lehtinen studied the expressiveness and succinctness of alternating GFG automata in~\cite{BL19}, showing that they
\begin{itemize}
\item are not more succinct than DFAs on finite words,
\item are as expressive as deterministic ones of the same acceptance condition on infinite words,
\item and can be determinised with a $2^{\theta(n)}$ size blowup for the B\"uchi and coB\"uchi conditions.
\end{itemize}

Many questions about GFG alternating automata were left open, in particular
whether there exists a~doubly exponential gap between alternating GFG and deterministic automata, and the complexity of deciding whether an~alternating parity automaton is GFG. We pursue the study of these questions, and obtain a~deeper understanding of the GFG realm.

\paragraph*{\large Contributions}
\Subject{Succinctness of alternating GFG automata}
We show that there is a single exponential gap between alternating parity GFG automata and deterministic ones, thereby answering a~question left open in~\cite{BL19}.
However, we also show that alternating GFG automata can present exponential succinctness compared to both nondeterministic and universal GFG automata. This means that alternating GFG automata can be used to reduce the complexity of solving some games with complex acceptance conditions.

\Subject{Recognising GFG automata}

We show that deciding whether an alternating automata on finite words or a weak alternating automata on infinite words is GFG is in \ptime.

For more general acceptance conditions such as parity, we rely on the two\=/token game $\Gt$ introduced in~\cite{BK18}. Bagnol and Kuperberg showed in~\cite{BK18} that this game characterises GFGness for nondeterministic B\"uchi automata, in the sense that the first player has a winning strategy in $\Gt$ if and only if the automaton is GFG. They conjectured that this result holds in general for parity conditions, which would provide a \ptime procedure to decide whether a given nondeterministic parity automaton with a fixed index is GFG.
We lift this characterisation to alternating automata: we define an alternating version of $\Gt$, which can still be solved in \ptime for automata of fixed index. Moreover, we prove that $\Gt$ characterises GFG alternating automata, provided that the conjecture holds for nondeterministic automata.
We then prove that the conjecture indeed holds for nondeterministic \textit{coB\"uchi} automata, taking a step towards a general solution. This immediately provides a new \ptime algorithm for recognising GFGness in nondeterministic coB\"uchi automata. It is simpler than the one from~\cite{KS15}, which involves several games and intermediate modifications of the input automaton.
Falling short of giving a~\ptime algorithm in the general case, we give an \exptime upper bound to the problem of deciding whether an~alternating parity automaton is GFG, matching the known upper bound for recognising nondeterministic GFG automata. 

We also study the complexity of deciding ``half\=/GFGness'', i.e.,\ whether the nondeterminism (or universality) of an automaton is GFG. This property guarantees that composition with games preserves the winner for one of the players. We show that already on finite words, this problem is \pspace-hard, and it is in \exptime for alternating B\"uchi automata.
This shows that a \ptime algorithm for deciding GFGness must, as in the case of finite word automata and weak automata, exploit the subtle interplay between nondeterminism and universality, and cannot be reduced to checking independently whether each of them is~GFG.

\Subject{Roadmap} We begin with some definitions, after which, in~\cref{sec:Alternating-behaviour}, we define alternating GFG automata, study their succinctness and the complexity of deciding GFGness of the nondeterminism within an alternating automaton.  \Cref{sec:Determinisation} provides a single-exponential determinisation procedure for alternating GFG parity automata. \Cref{sec:deciding} shows that GFGness of alternating parity automata is in \exptime and works towards a \ptime algorithm. In particular, it provides such an algorithm for weak automata. Finally,~\cref{sec:G2-coBuchi} shows that the $\Gt$ conjecture holds for coB\"uchi automata. Throughout the paper, we provide high-level proof sketches, with detailed technical developments 
in the appendix.

%% file: 2_prelim.tex
\section{Preliminaries}
\label{sec:Preliminaries}

\Subject{Words and automata}
An \emph{alphabet} $\Sigma$ is a finite nonempty set of letters. A finite (resp.\ infinite) \emph{word} $u=u_0 \ldots u_k\in \Sigma^{*}$ (resp.\ $w=w_0 w_1\ldots\in \Sigma^{\omega}$) is a finite (resp.\ infinite) sequence of letters from $\Sigma$. 
A \emph{language} is a set of words, and the empty word is written $\epsilon$. We denote a set $\{i,\ldots,j\}$ of integers by $[i,j]$.
 
An \emph{alternating word automaton} is a~tuple $\A=(\Sigma,Q,\iota,\delta,\alpha)$, where: $\Sigma$ is an alphabet; $Q$ is a finite nonempty set of states; $\iota\in Q$ is an~initial state; $\delta\colon Q\times \Sigma \to \BP(Q)$ is a transition function where $\BP(Q)$ is the set of positive Boolean formulas (\emph{transition conditions}) over $Q$; and $\alpha$, on which we elaborate below, is either an acceptance condition or a transition labelling on top of which an acceptance condition is defined.
For a state $q\in Q$, we denote by $\A^q$ the automaton that is derived from $\A$ by setting its initial state $\iota$ to $q$. 

An automaton $\A$ is nondeterministic (resp.\ universal) if all its transition conditions are disjunctions (resp.\ conjunctions), and it is deterministic if all its transition conditions are just states. We represent the transition function of nondeterministic and universal automata as $\delta\colon Q\times \Sigma\to 2^Q$, and of a deterministic automaton as $\delta\colon Q\times \Sigma\to Q$. A \emph{transition} of an automaton is a triple $(q,\letter,q')\in Q{\times}\Sigma{\times} Q$, sometimes also written $\trans{q}{\letter}{q'}$.

We denote by $\widehat{\delta}\subseteq \BP(Q)$ the set of all subformulas of formulas in the image of $\delta$, i.e., all the Boolean formulas that ``appear'' somewhere in the transition function of $\A$. 

\Subject{Acceptance conditions}
There are various acceptance (winning) conditions, defined with respect to the set of transitions\footnote{Acceptance is defined in the literature with respect to either states or transitions; for technical reasons we prefer to work with acceptance on transitions.} that a path of $\A$ visits infinitely often. (Notice that a transition condition allows for many possible transitions.)
We later formally define acceptance of a word $w$ by $\A$ in terms of games, and consider a path of $\A$ on a word $w$ as a play in that game. For nondeterministic automata, a ``run'' coincides with a ``path''.

Some of the acceptance conditions are defined on top of a labelling of the transitions rather than directly on the transitions. In particular, in the parity condition, we have $\alpha\colon Q\times\Sigma\times Q \to \Gamma$, where $\Gamma\subseteq\Nat$ is a finite set of priorities and a path is accepting if and only if the highest priority seen infinitely often on it is even.

The B\"uchi and coB\"uchi conditions are special cases of the parity condition with $\Gamma=\{1,2\}$ and $\Gamma=\{0,1\}$, respectively. When speaking of B\"uchi and coB\"uchi automata, we often refer to $\alpha$ as the set of ``accepting transitions'', namely the transitions that are mapped to $2$ in the B\"uchi case and to $0$ in the coB\"uchi case.
The weak condition is a special case of both the B\"uchi and coB\"uchi conditions, in which every path eventually remains in the same priority.

The Rabin and Streett conditions are more involved, yet defined directly on the set $T$ of transitions. A Rabin condition is a set $\{\pair{B_1,G_1}, \pair{B_2,G_2},\ldots, \pair{B_k,G_k}\}$, with $B_i, G_i \subseteq T$,
and a path $\rho$ is accepting iff for some $i\in\{1,\ldots,k\}$, we have that the set $inf(\rho)$ of transitions that are visited infinitely often in $\rho$ satisfies ($inf(\rho) \cap B_i = \emptyset$ and $inf(\rho) \cap G_i \neq \emptyset$). A Streett condition is dual: a set $\{\pair{B_1,G_1}, \pair{B_2,G_2},
\ldots, \pair{B_k,G_k}\}$, with $B_i, G_i \subseteq Q$, whereby a~path $\rho$ is accepting iff for all $i\in\{1,\ldots,k\}$, we have ($inf(\rho) \cap B_i = \emptyset$ or $inf(\rho) \cap G_i \neq \emptyset$).

\Subject{Sizes and types of automata}
The size of $\A$ is the maximum of the alphabet size, the number of states, the transition function length, which is the sum of the transition condition lengths over all states and letters, and the acceptance condition's index, which is $1$ for weak, B\"uchi and coB\"uchi, $|\Gamma|$ for parity, and $k$ for Rabin and Street.

We sometimes abbreviate automata types by three-letter acronyms in $\{$D, N, U, A$\} \times \{$F, W, B, C, P, R, S$\} \times \{$A,W$\}$. The first letter stands for the transition mode, the second for the acceptance-condition, and the third indicates that the automaton runs on finite or infinite words. For example, DPW stands for a deterministic parity automaton on infinite words.

\Subject{Games and strategies}
Some of our technical proofs use standard concepts of an arena, a game, a winning strategy, etc\ldots For the sake of completeness, we provide precise mathematical definitions of these objects in \cref{ap:Preliminaries}. Here we will just overview the involved concepts.

First, we work with two\=/player games of perfect information, where the players are Eve and Adam. These games are played on graphs (called arenas). Most of the considered games are of infinite duration and their winning condition is expressed in terms of the infinite sequences of edges taken during the play. We invoke results of determinacy (one of the players has a winning strategy), as well as of \emph{positional determinacy} (one of the players has a strategy that depends only on the last position of the play).

\Subject{Model-checking games}

To represent the semantics of an alternating automaton $\A$, we treat the Boolean formulas that appear in the transition conditions of $\A$ as games. More precisely, given a letter $\letter\in\Sigma$ we represent the transition conditions $q\mapsto\delta(q,\letter)\in\BP(Q)$ as the \emph{one-step arena} over $\letter$. A play over this arena begins in a state $q\in Q$; then players go down the formula $\delta(q,a)$ with Eve resolving disjunctions and Adam resolving conjunctions; and finally they reach an atom $q'\in Q$ and the play stops. This means that a play over the one-step arena over $\letter$ results in a transition of the form $(q,\letter,q')$.

The language $L(\A)$ of an alternating automaton $\A$ over an alphabet $\Sigma$ is defined via the \emph{model-checking game}. A configuration of this game is a state $q$ of $\A$, starting at $\iota$. In the $i$th round, starting from state $q_i$, the players play the game over the one-step arena over $w_i$, resulting in a transition $\trans{q}{w_i}{q_{i+1}}$. The acceptance condition of $\A$ becomes the winning condition of this game. $\A$ \emph{accepts} a word $w\in\Sigma^\omega$ if Eve has a winning strategy in this game. 

For technical convenience, we define (in \cref{ap:Preliminaries}) the model-checking game in terms of a \emph{synchronised product} of the word $w$ (treated as an infinite graph) and the automaton $\A$. Synchronised products turn out to be useful in the analysis of various games presented in this paper and will be used throughout the technical versions of the proofs, in the appendix.

\begin{definition}
\label{def:auto-compl}
Given an alternating automaton $\A$, we denote by $\overline{\A}$ the \emph{dual} automaton: it has the same alphabet, set of states, and initial state. Its transition conditions $\delta_{\overline{\A}}(q,\letter)$ are obtained from those of $\A$ by replacing each disjunction ${\lor}$ with conjunction ${\land}$ and vice versa. Its acceptance condition is the dual of $\A's$ condition. (In parity automata, all priorities are increased by $1$.)
\end{definition}

\Subject{Boxes}
Another technical concept that we use is that of \emph{boxes}. They can be defined with respect to the synchronised product, see page~\pageref{def:boxes}, but also directly based on transition conditions. Consider an alternating automaton $\A$ and a letter $\letter\in\Sigma$. Moreover, fix a strategy $\stratE$ of Eve that resolves disjunctions in all the transition conditions $\delta(q,\letter)$. Now, the \emph{box} of $\A$, $\letter$, and $\stratE$, denoted $\beta(\A,\letter,\stratE)$ is the subset of $Q\times\{\letter\}\times Q$ that contains $(q,\letter,q')$ if there is a play consistent with $\stratE$ on $\delta(q,\letter)$ that reaches the atom $q'$ of the formula. By $\boxes_{\A,\letter}$ we denote the set of all boxes of $\A$ and $\letter$, while  $\boxes_\A$ denotes the union $\bigcup_{\letter\in\Sigma} \boxes_{\A,\letter}$.

\begin{definition}
\label{def:path-in-boxes}
Given a~sequence of boxes $\pi=b_0,b_1,\ldots$ of an automaton $\A$ and a~path $\rho=(q_0,\letter_0,q_1),(q_1,\letter_1,q_2),\ldots$, we say that $\rho$ is a~\emph{path of $\pi$} if for every $i$ we have $(q_i,\letter_i,q_{i+1})\in b_i$.
The sequence $\pi$ is said to be~\emph{universally accepting} if every path in $\pi$ is accepting in $\A$.
\end{definition}

Intuitively, a sequence of boxes $\pi$ as above represents a particular strategy $\strat$ of Eve in the model-checking game over the word $w=a_0,a_1,\ldots$ In that case, a path of $\pi$ corresponds to a possible play of this game consistent with $\strat$.

%% file: 3_alt-gfg.tex
\section{Good-For-Games Alternating Automata}
\label{sec:Alternating-behaviour}

Good-for-games (GFG) nondeterministic automata are automata in which the nondeterministic choices can be resolved without looking at the future of the word. For example, consider an automaton that consists of a nondeterministic choice between a component that accepts words in which $a$ occurs infinitely often and a component that accepts words in which $a$ occurs finitely often. This automaton accepts all words but is not GFG since the nondeterministic choice of component cannot be resolved without knowing the whole word.

To extend this definition to alternating automata, we must look both at its nondeterminism and universality and require that both can be resolved without knowledge of the future. The following letter games capture this intuition.

\begin{definition}[Letter games~\cite{BL19}]
\label{def:LetterGames}
Given an alternating automaton $\A$, Eve's letter game proceeds at each turn from a state $q$ of $\A$, starting from the initial state of $\A$, as follows:
\begin{itemize}
\item Adam chooses a letter $\letter$,
\item Adam and Eve play on the one-step arena over $\letter$ from $q$ to a new state $q'$, where Eve resolves disjunctions and Adam conjunctions.
\end{itemize}
A play of the letter game thus generates a word $w$ and a path $\rho$ of $\A$ on $w$. Eve wins this play if either $w\notin L(\A)$ or $\rho$ is accepting in $\A$.

Adam's letter game is similar, except that Eve chooses letters and Adam wins if either $w\in L(\A)$ or the path $\rho$ is rejecting.
\end{definition}
A more formal definition is given in \cref{ap:Alternating}.

\begin{definition}[GFG automata~\cite{BL19}]
An automaton $\A$ is \emph{$\EGFG$} if Eve wins her letter game; it is \emph{$\AGFG$} if Adam wins his letter game. Finally, $\A$ is \emph{GFG} if it is both $\EGFG$ and $\AGFG$.
\end{definition}

As shown in \cite[Theorem~8]{BL19}, an automaton $\A$ is GFG if and only if it is indeed ``good for playing games'', in the sense that its product with every game whose winning condition is $L(\A)$ preserves the winner of the game.

\subsection{Alternating GFG vs. Nondeterministic and Universal Ones}


We show in this section that alternating GFG automata can be more succinct than both nondeterministic and universal GFG automata.

\begin{lemma}
\label{lem:Cn_family}
There is a family $(\C_n)_{n\in\Nat}$ of alternating GFG $\{0,1,2\}$-parity automata of size linear in $n$ over a fixed alphabet, such that every nondeterministic GFG parity automaton and universal GFG parity automaton for $L(\C_n)$ is of size $2^{\Omega(n)}$.
\end{lemma}

\ProofSketch
We use the succinctness result from \cite[Thm. 1]{KS15}, stating that there exists a family $\A_n$ of NCW-GFG with size linear in $n$, such that any DPW for $L(\A_n)$ has exponential size.
Combining $\A_n$ and its dual into a single alternating automaton gives us the wanted result. See \cref{ap:Cn} for a detailed construction.
\end{proof}

\subsection{Deciding Half-GFGness}

In order to decide GFGness, it is enough to be able to decide the $\EGFG$ property on the automaton and its dual.
A natural first approach is therefore to study the complexity of deciding whether an APW is $\EGFG$.
Yet, we will show that already on finite words, this problem is more difficult than deciding GFGness.

\begin{lemma}
\label{lem:EGFG_PSPACE}
Deciding whether an AFA is $\EGFG$ is \pspace-hard.
\end{lemma}

\ProofSketch
We reduce from NFA universality: starting from an NFA $\A$, we build an AFA $\B$ based on the dual of $\A$, with an additional non-GFG choice to be resolved by Eve. This AFA $\B$ is $\EGFG$ if and only if $L(\B)=\emptyset$, which happens if and only if $L(\A)=\Sigma^*$.  We crucially use the fact that $\B$ is not necessarily $\AGFG$. See Appendix \ref{ap:pspace} for a detailed construction.
\end{proof}

For B\"uchi automata, and so in particular for finite words, we can give an \exptime algorithm for this problem.

\begin{lemma}
\label{lem:EGFG_EXPTIME}
Deciding whether an ABW is $\EGFG$ is in \exptime.
\end{lemma}

\begin{proof}
It is shown in~\cite[Lemma~23]{BL19} that removing alternation from an ABW $\A$ using the breakpoint construction~\cite{MH84} yields an NBW such that if $\A$ is $\EGFG$ then $\B$ is GFG. Moreover, it is straightforward to show that the converse also holds, i.e., if $\B$ is GFG then $\A$ is $\EGFG$, since playing Eve's letter game in $\B$ is more difficult for Eve than playing it in $\A$.
This means that starting from an ABW $\A$, we can build an exponential size NBW $\B$ via breakpoint construction, and test whether $\B$ is GFG via the algorithm from~\cite{BK18}, in time polynomial with respect to $\B$. Overall, this yields an \exptime algorithm deciding whether $\A$ is $\EGFG$.
\end{proof}

In contrast, we will show in \cref{sec:deciding} that deciding GFGness for AFA and AWW is in \ptime, and conjecture that the same is true for APW of every fixed index.

%% file: 4_determ.tex

\section{Determinisation of Alternating GFG Parity Automata}
\label{sec:Determinisation}

In this section we provide a procedure that, given an~alternating GFG parity automaton, produces an equivalent deterministic parity automaton with singly exponentially many states. To do so, we first provide an alternation-removal procedure for Rabin automata that preserves GFG status. Then, we apply this procedure to both the input automaton and its complement and use the GFG strategies in these two automata to determinise the input. Our proofs, in~\cref{ap:Determinisation}, rely on some analysis of when GFG strategies can use the history of the word, rather than the whole play, and on the memoryless determinacy of Rabin games.

Our method for going from alternating to nondeterministic automata is similar to that of Dax and Klaedtke~\cite{DK08}: they take a nondeterministic automaton that recognises the universally-accepting words in $(\boxes_A)^\omega$ and add nondeterminism that upon reading a letter $\letter\in \Sigma$ chooses a box in $\boxes_A$ over \letter. Yet in our approach, in order to guarantee that the outcome preserves GFGnesss, the intermediate automaton is deterministic.

\subsection{Alternation Removal in GFG Rabin Automata}


\begin{restatable}{theorem}{thmexpgfgdealt}
\label{thm:exp-gfg-dealt}
Consider an alternating Rabin (resp.\ parity) automaton $\A$ with $n$ states and index $k$. There exists a nondeterministic parity automaton $\boxA$ with $2^{O(nk \log nk)}$ (resp. $2^{O(n \log n)}$) states that is equivalent to $\A$ such that if $\A$ if GFG then $\boxA$ is also GFG.
\end{restatable}

In~\cref{sec:deciding}, where we discuss decision procedures, we will show that $\boxA$ is GFG \textit{exactly} when $\A$ is GFG. For now, the rest of this section is devoted to the proof of~\cref{thm:exp-gfg-dealt}, of which a detailed version can be found in~\cref{app:alt-rem-gfg-rabin}. 

\begin{restatable}{lemma}{lempropofd}
\label{lem:prop-of-d}
Consider an alternating Rabin (resp.\ parity) automaton $\A$ with $n$ states and index $k$. Then there exists a deterministic parity automaton $\B$ with $2^{O(nk \log nk)}$ (resp. $2^{O(n \log n)}$) states over the alphabet $\boxes_\A$ that recognises the set of universally-accepting words for $\A$. If $\A$ is a B\"uchi automaton, then $\B$ can also been taken as B\"uchi, and in general the parity index of the automaton $\B$ is linear in the number of transitions~of~$\A$.
\end{restatable}

\ProofSketch
We construct the automaton $\B$ by determinising and complementing a nondeterministic Streett (resp.\ parity or coB\"uchi) automaton over the alphabet $\boxes_\A$ that recognises the complement of the set of universally-accepting words for $\A$, that is, an~automaton that guesses a path that is not accepting, and has the dual acceptance condition to $\A$. 
\end{proof}

\NotNeeded{
\begin{proof}
Notice that it is easy to construct a nondeterministic Streett (resp.\ parity) automaton $\S$ over the alphabet $\boxes_\A$ that recognises the complement of the set of universally-accepting words for $\A$---it is enough to guess a path that is not accepting, and have the acceptance condition that is the dual of $\A$'s condition. 
Formally, for an alternating Rabin (resp.\ parity) automaton $\A=\tuple{\Sigma, Q, \iota, \delta, \alpha}$, we define the nondeterministic Street (resp.\ parity) automaton $\S=\tuple{\boxes_\A, Q, \iota, \delta_\S, \overline{\alpha}}$, where $\overline{\alpha}$ is the dual of $\alpha$ and $\delta_\S$ is defined as follows.
For every states $q,q'\in Q$ and box $\beta\in\boxes_\A$, we have $q'\in\delta_\S(q,\beta)$ iff $\tuple{q,q'}\in\beta$.

Now, one can translate $\S$ to an equivalent deterministic parity automaton $\B'$ with $2^{O(nk \log nk)}$ states \cite{Pit07} (resp.\ $2^{O(n \log n)}$ states~\cite{CZ12,SV14}), and then complement the acceptance condition of $\B'$, getting the required automaton $\B$.

Since nondeterministic coB\"uchi automata can be determinised into deterministic coB\"uchi automata, if $\A$ is a B\"uchi automaton, so is $\B$.
\end{proof}
}


We now build the automaton $\boxA$ of \cref{thm:exp-gfg-dealt}. It is the same as the automaton $\B$ of \cref{lem:prop-of-d}, except that the alphabet is $\Sigma$ and the transition function is defined as follows: For every state $p$ of $\B$ and $\letter\in\Sigma$, we have $\delta_{\boxA}(p,\letter):= \bigcup_{\beta\in\boxes_{\tuple{\A,\letter}}} \delta_{\B}(p,\beta)$.

In other words, the automaton $\boxA$ reads a~letter $\letter$, nondeterministically guesses a~box $\beta\in\boxes_{\A,\letter}$, and follows the transition of $\B$ over $\beta$. Thus, the runs of $\boxA$ over a~word $w\in\Sigma^\omega$ are in bijection with sequences of boxes $(\beta_i)_{i\in\Nat}$ such that $\beta_i\in\boxes_{\A,w_i}$ for $i\in\Nat$.


Fix an infinite word $w\in\Sigma^\omega$. Our aim is to prove that $w\in L(\A)\Leftrightarrow w\in L(\boxA)$. 

\begin{restatable}{lemma}{lempositionalstrategiesandruns}
\label{lem:PositionalStrategiesAndRuns}
There exists a bijection between positional strategies of Eve in the acceptance game of $\A$ over $w$ and runs of $\boxA$ over $w$. Moreover, a strategy is winning if and only if the corresponding run is accepting.
Then $L(\A)=L(\boxA)$.
\end{restatable}

\NotNeeded{
\begin{proof}
Consider a run of $\boxA$ over $w$, and observe that it corresponds to a sequence of boxes $\beta_0,\ldots$. Notice that each box $\beta_i$ corresponds to Eve's choices in $\A$ over $w_i$, and therefore provides a positional strategy for Eve in the one-step arena $R_{w_i}\times \A$. The sequence of these choices provides a positional strategy for Eve in $R_w\times \A$. 

Dually, given a positional strategy for Eve in $R_w\times \A$, one can extract a sequence of strategies for Eve in the one-step arenas $R_{w_i}\times \A$, and each of them corresponds to a box $\beta_i$. \cref{prop:choice-to-strat} shows that each path in $\beta_0,\ldots$ corresponds to a play consistent with the constructed strategy and vice versa: each play gives rise to a path.

Now, a run is accepting if and only if the sequence of boxes is universally accepting, which means exactly that all the plays consistent with the corresponding strategy are winning.
\end{proof}
}


\begin{remark}
The above alternation-removal procedure fails for alternating Streett automata~$\A$: since Streett games are not positionally determined for Eve, the acceptance game of $\A$ over a word $w$ is not positionally determined for Eve.
\end{remark}

\begin{restatable}{lemma}{lemgfgpreservation}
\label{lem:GFG-preservation}
For an alternating $\EGFG$ Rabin automaton $\A$, the automaton $\boxA$ is~GFG.
\end{restatable}

Intuitively, this is because the construction of $\boxA$ preserves the nondeterminism of $\A$. 

\NotNeeded{
\begin{proof}
Let $\stratE$ be a positional winning strategy for Eve in her expanded letter game for $\A$ (over the arena $R^\ast_{A,\Sigma}\times \A$). The proof is based on the construction of the function $\stratE'\colon \Sigma^+\rightarrow \boxes_\A$, see the paragraph before \cref{def:univ-acc-box}.

Consider the following way of resolving the nondeterminism of $\boxA$: after reading $w\in\Sigma^\ast$, when the next letter $\letter\in\Sigma$ is provided, the automaton moves to the state $\delta_\D(p,\beta_{w\letter})$ where $\beta_{w\letter}=\stratE'(w\letter)$. Consider an infinite word $w\in L(\A)$ and let $\beta_0,\ldots$ be the sequence of boxes used to construct the run of $\boxA$ over $w$. Lemma~\ref{lem:strat-to-win-boxes} implies that this sequence is universally accepting and therefore, the constructed run of $\D$ must also be accepting.
\end{proof}
}

\subsection{Single-Exponential Determinisation}
\label{ssec:exp-det-of-alt}

The aim of this section is to prove the following determinisation theorem; See~\cref{ap:ssec:exp-det-of-alt} for a detailed proof.

\begin{restatable}{theorem}{thmdet}
If $\A$ is an alternating parity GFG automaton then there exists a deterministic parity automaton $\D$ that recognises the same language and has size at most exponential in the size of $\A$. Moreover, the parity index of $\D$ is the same as that of $\A$.
\end{restatable}

\begin{remark}
\cref{thm:exp-gfg-dealt} and~\cite[Theorem~4]{BKKS13} together give an exponential deterministic parity (Rabin) automaton for $L(\A)$. However, the index of $\A$ might not be preserved. On the other hand, from~\cite[Theorem~19]{BL19} we know that there exists a~deterministic parity automaton equivalent to $\A$ with the same index, but it might have more than exponentially many states.
\end{remark}

Observe that \cref{thm:exp-gfg-dealt} can be applied both to $\A$ and its dual. Therefore, we can fix a~pair of nondeterministic GFG parity automata $\boxA$ and $\boxAco$ that recognise $L(\A)$ and $L(\A)^\mathrm{c}$ respectively and are both of size exponential in $\A$.
We use the automata $\A$, $\boxA$, and $\boxAco$ to construct two auxiliary games.\\

The game $G(\A)$ proceeds from a configuration consisting of a pair $(p,q)$ of states from $\boxAco$ and $\A$ respectively, starting from their initial states, as follows:
\begin{itemize}
\item Adam chooses a letter $\letter\in \Sigma$;
\item Eve chooses a transition $\trans{p}{\letter}{p'}$ in $\boxAco$;
\item Eve and Adam play on the one-step arena over $\letter$ from $q$ to a new state $q'$.
\end{itemize}

A play in $G(\A)$ consists of a run $\rho$ in $\boxAco$ and a path $\rho'$ in $\A$. It is winning for Eve if either $\rho$ is accepting in $\boxAco$ (in which case $w\notin L(\A)$), or $\rho'$ is accepting in $\A$.\\

If $\A$ is $\EGFG$ and $\boxAco$ is GFG, Eve has a winning strategy in $G(\A)$ consisting of building a run in $\boxAco$ using her GFG strategy in $\boxAco$ and a path in $\A$ using her $\EGFG$ strategy in $\A$. This guarantees that if $w\in L(\A)$ then the path in $\A$ is accepting, and otherwise the run in $\boxAco$ is accepting.

We then argue that as the winning condition of $G(\A)$ is a Rabin condition, Eve also has a winning strategy that is positional in $\A$, that is, which only depends on the history of the word and the current position. See~\cref{ap:Determinisation} for details.

\begin{remark}
There is some magic here: both the GFG strategies of Eve in $\A$ and in $\boxAco$ may require exponential memory, yet, when she needs to satisfy the disjunction of the two conditions, no more memory is needed. In a sense, the states of $\A$ provide the memory for $\boxAco$ and the states of $\boxAco$ provide the memory for $\A$.
\end{remark}

The game $G'(\A)$ is similar, except that Adam is given control of $\boxA$ and Eve is in charge of letters. This time Adam wins a play consisting of a run of $\boxA$ and a path in $\A$ if either the path of $\A$ is rejecting  or the run of $\boxA$ is accepting.

Accordingly, if $\A$ is GFG, then he can win by using the GFG strategy in $\boxA$ and the $\AGFG$ strategy in $\A$. Then if $w\in L(\A)$, the run in $\boxA$ is accepting, and otherwise the path of $\A$ is rejecting.
As before, he also has a positional winning strategy in $G'(\A)$.\\

\NotNeeded{
\todo{begin into appendix}
Now do the same with $\A$ and $\boxA$ for Adam: define $G$ as $\big(R_{E,\Sigma}\times \overline{\boxA}\big)\times \A$, where $\overline{\boxA}$ is the automaton $\boxA$ where the transitions are turned from nondeterministic to universal, i.e,\ we replace $\lor$ with $\land$.

Again, in a~round of $G$ from a~position $(v,p,q)$: Eve plays a letter $\letter$; Adam resolves nondeterminism of $\boxA$; then they both resolve the choices in $\A$. Let Adam win $G$ if either the play of $\A$ is rejecting or the run of $\boxA$ is accepting. Again we can ensure that Adam has a~winning strategy in $G$, because both automata are GFG: he uses the GFG strategy of $\boxA$ and the $\AGFG$ strategy over $\A$. If the word given by Eve belongs to $L(\A)$ then Adam wins by producing an accepting run of $\boxA$, otherwise he wins by refuting an accepting run of $\A$. Let $\stratA$ be his positional winning strategy in that game.

\todo{end}
}

We are now ready to build the deterministic automaton from a GFG APW $\A$, using positional winning strategies $\stratE$ and $\stratA$ for Eve and Adam in $G(\A)$ and $G'(\A)$, respectively.
 
Let $\D$ be the automaton with states of the form $(q,p_1,p_2)$, with $q$ a~state of $\A$, $p$ a~state of $\boxA$ and $p'$ a state of $\boxAco$. 
A transition of $D$ over $\letter$ moves to $(q',p_1',p_2')$ such that moving from $(q,p_1)$ to $(q',p_1)$ is consistent with $\stratA$; and moving from $(q,p_2)$ to $(q',p_2')$ is consistent with $\stratE$.
The acceptance condition of $\D$ is inherited from $\A$.

\NotNeeded{
\todo{start appendix} 
When reading a~letter $\letter$ in such a~state, the following computations are performed:
\begin{enumerate}
\item We simulate the choices made by $\stratE$ in $G'(\A)$ upon obtaining $\letter$ from Adam. This way we know how to resolve nondeterminism of $\boxAco$ and what to do with disjunctions inside $\A$.
\item We simulate the choices made by $\stratA$ in $G(\A)$ upon obtaining $\letter$ from Eve. This way we know how to resolve nondeterminism in $\boxA$ and what to do with conjunctions of $\A$.
\item In the end we proceed to a new state of $\A$ and resolved nondeterminism of both $\boxA$ and $\boxAco$.
\end{enumerate}

Let the acceptance condition of $\D$ be inherited from $\A$.
\todo{end appendix}
}

\begin{restatable}{lemma}{lemdeteq}\label{lem:det-eq}
For a GFG APW $\A$ and $\D$ built as above, $L(\A)=L(\D)$.
\end{restatable}

\NotNeeded{
\todo{start appendix}

\begin{proof}
Take a word $w\in\Sigma^\omega$. First assume that $w\in L(\A)$. Eve cannot win a play of the game $G'$ with the letters played in $R_{A,\Sigma}$ coming from $w$ using by the first disjunct of her winning condition, since $L(\boxAco)=L(\bar\A)$. Thus, all the plays over $w$ consistent with her winning strategy $\stratE$ in $\G'$ must guarantee that the constructed path of $\A$ is accepting. Thus, the run of the automaton $\D$ over $w$ is accepting.

Now assume that $w\notin L(\A)$. Dually, no play of the game $G$ with the letters coming from $w$ can produce an accepting run of $\boxA$ over $w$. Thus, the strategy $\stratA$ guarantees that the sequence of visited states of $\A$ is rejecting. Thus, the run of $\D$ over $w$ must be rejecting.
\end{proof}

\todo{end appendix}
}

\begin{remark}
The above construction does not work for an alternating GFG Rabin automaton~$\A$, since we need to remove alternations from both $\A$ and its dual. Although we know how to remove alternations from $\A$ with a singly-exponential size blowup while preserving GFGness, we do not know how to do it to the dual of $\A$, which is a Streett automaton.
\end{remark}

%% file: 5_deciding.tex
\section{Deciding GFGness of Alternating Automata}
\label{sec:deciding}

We first use the development of the last section to show that deciding whether an APW is GFG is in \exptime. This matches the best known upper bound for the same problem on NPW. We then consider how to improve this upper bound by characterising GFGness with a polynomially solvable game. In particular, we show that if the token game $\Gt$ known to characterise GFGness for NBW, can be shown to also characterise GFGness for  \textit{nondeterministic} parity automata, as previously conjectured in \cite{BK18}, then it also characterises GFGness for \textit{alternating} parity automata. In the special case of AWW (or AFAs), we show that this token game indeed characterises GFGness, and can be decided polynomially.

\subsection{GFGness of Alternating Parity Automata is in \exptime}
\label{ssec:deciding-exptime}

The main result of this section is the following theorem; its proof is in~\cref{app:exptime-gfg-alt}.

\begin{theorem}
\label{thm:exp-time-gfg-alt}
There exists an~\exptime algorithm that takes as input an alternating parity automaton $\A$ and decides whether $\A$ is GFG.
\end{theorem}

A complete proof of this result is given in \cref{app:exptime-gfg-alt}. The idea is to construct the (exponential size) NPWs $\boxA$ and $\boxAco$ for $L(\A)$ and $L(\A)^\mathrm{c}$ respectively. We observe the following reciprocal of \cref{lem:GFG-preservation}.

\begin{lemma}
\label{lem:GFG-B-to-A}
If $\boxA$ is GFG then $\A$ is $\EGFG$.
\end{lemma}

\begin{proof}
Assume that $\boxA$ is GFG and consider a~strategy witnessing this. Such a~strategy can be easily turned into a~function $\stratE'\colon \Sigma^+\rightarrow \boxes_\A$ that, given a~word $w\in L(\A)$ produces a~universally accepting word of boxes of $\A$. Now, due to the definition of a box, each such box defines a~positional strategy of Eve in the respective one-step game. This allows us to construct a~winning strategy of Eve in the letter game over $\A$.
\end{proof}

Thus, $\A$ is GFG if and only if both $\boxA$ and $\boxAco$ are GFG. To decide this, we consider a~game $G''$ where Adam plays letters and Eve produces runs of the automata $\boxA$ and $\boxAco$ in parallel. The winning condition of $G''$ requires that at least one of the constructed runs must be accepting.

Now, each sequence of letters given by Adam belongs either to the language of $\boxA$ or to $\boxAco$ and therefore, a winning strategy of Eve in $G''$ must comprise of two strategies witnessing GFGness of both $\boxA$ and $\boxAco$. Dually, if both $\boxA$ and $\boxAco$ are GFG then Eve wins $G''$ by playing the two strategies in parallel.

A careful analysis of the winning condition of $G''$ shows that solving it is in \exptime.

\subsection{Towards a Polynomial Procedure}
\label{ssec:G2-for-alternating}

While the letter games characterise whether an automaton is GFG, solving these games is not as easy as one could hope, as the winning condition depends on whether the played word is in the language. The naive solution is to use a deterministic automaton to recognise whether the played word is in the language; however the cost of determinisation is, in the case of alternating automata, doubly exponential. \cref{thm:exp-time-gfg-alt} already improves on this by giving a single exponential procedure.


The hope for further improving on this upper bound is to find an alternative characterisation of GFGness, based on polynomially solvable games. So far, this approach has been successful in the case of nondeterministic B\"uchi automata~\cite{BK18}: a nondeterministic B\"uchi automaton is GFG if and only if Eve wins the game $\Gt$ in which Adam chooses letters while Eve builds a run in the automaton, as in the letter game, but, in addition, Adam also has to build two runs, of which at least one should witness that the word is in the language. This game is polynomially solvable  as the arena is just the product of the alphabet and three copies of the automaton, and the winning condition is a simple Boolean combination of B\"uchi conditions. Asking Adam to just build one accepting run would make the game too easy for Eve who could use the information from Adam's run to build her own run, see \cite[Lemma 8]{BK18}. 

We describe below a version of the $\Gt$ game suited to alternating automata.

\begin{definition}[The two-token game]
Given an alternating parity automaton $\A$, we define the \emph{two-token game} $\Gt(\A)$. 
A configuration $(p,q_1,q_2)\in Q^3$ of $\Gt(\A)$ consists of three states of $\A$, one for Eve's token, and two for Adam's tokens. The initial configuration is $(\iota,\iota,\iota)$. A~turn starting in $(p,q_1,q_2)$ proceeds as follows:
\begin{itemize}
\item Adam picks a letter $\letter\in\Sigma$;
\item Eve and Adam play the one-step game over $\delta_\A(p,\letter)$ in $\A$ and build a transition $\trans{p}{\letter}{p'}$;
\item Eve and Adam play the one-step game over $\delta_{\bar{\A}}(q_1,\letter)$ in $\bar{\A}$ and build a transition $\trans{q_1}{\letter}{q_1'}$;
\item Eve and Adam play the one-step game over $\delta_{\bar{\A}}(q_2,\letter)$ in $\bar{\A}$ and build a transition $\trans{q_2}{\letter}{q_2'}$;
\item The new configuration is $(p',q_1',q_2')$.
\end{itemize}

A play consists of the resulting three infinite paths $(\rho_E,\rho_A,\rho_A')$ and is winning for Eve if either $\rho_E$ is accepting or $\rho_A$ and $\rho_A'$ are both rejecting.
\end{definition}

Notice that the roles of the players are swapped in the games from $q_1$ and $q_2$: it is Adam who resolves disjunctions and Eve resolves conjunctions. 

It is easy to encode the above game as a~game over a finite graph, with labels of the form $(Q\times\Sigma\times Q)^3$, representing the three transitions taken in a given turn.

\begin{remark}
\label{rem:k-token-nondet}
Notice that if $\A$ is a nondeterministic automaton, then this game is just the two-token game from~\cite{BK18} in which Adam picks a letter, Eve chooses a transition for her token and Adam chooses transitions for his two tokens. In the nondeterministic case, we will also use the game $G_k(\A)$, in which Adam has $k$ tokens instead of two, see~\cite[Definition~9]{BK18}.
\end{remark}

\begin{theorem}[{\cite[Corollary~21]{BK18}}]
\label{thm:G2-buchi}
For all NBW $\A$, Eve wins $\Gt(\A)$ if and only if $\A$ is GFG.
\end{theorem}


\begin{conjecture}[\cite{BK18}]
\label{con:G2-to-GFG}
A nondeterministic parity automaton $\A$ is GFG if and only if Eve wins $\Gt(\A)$.
\end{conjecture}

In this section we show that if this conjecture holds, then the above generalisation of $\Gt$ also characterises GFGness for alternating automata, in the sense that then an alternating parity automaton $\A$ is GFG if and only if Eve wins both $\Gt(\A)$ and $\Gt(\bar\A)$.

Before we move on, we argue that the game $\Gt$ is more tractable than both the letter game and the approach from \cref{ssec:deciding-exptime}, as expressed by the following proposition.

\begin{restatable}{proposition}{proDecideGPoly}
\label{pro:decide-G2-polynomial}
Given an APW $\A$ of size $n$ with a fixed number~$d$ of priorities over an~alphabet $\Sigma$, the game $\Gt(\A)$ can be solved in time complexity $O(n^4)$. (More precisely, in $O(d^2(n^{3}|\Sigma| 2^{d^2\log d})^{1+o(1)})$.)
\end{restatable}

A~proof of this proposition boils down to a careful analysis of the size of $\Gt(\A)$ and ways to represent its winning condition, see \cref{app:G2-polynomial}.

The following lemma is direct: a GFG strategy of Eve in $\A$ can win $\Gt(\A)$ without even looking at the tokens moved by Adam, see \cref{app:GFG-to-G2}.

\begin{restatable}{lemma}{lemGfgToG}
\label{lem:gfg-to-G2}
If an alternating automaton $\A$ is GFG, then Eve wins both $\Gt(\A)$ and $\Gt(\bar \A)$.
\end{restatable}

Recall that in \cref{sec:Determinisation} we construct from $\A$ a nondeterministic parity automaton $\boxA$ which is GFG if and only if the nondeterminism in $\A$ is GFG, see \cref{lem:GFG-preservation} and \cref{lem:GFG-B-to-A}.

We now show that if Eve wins $\Gt(\A)$, then she also wins $\Gt(\boxA)$. Then, if \cref{con:G2-to-GFG} holds, it follows that $\boxA$, and therefore also $\A$, is GFG.

\begin{restatable}{proposition}{proAltToNd}
\label{pro:alt-to-nd}
For an~alternating parity automaton $\A$, if Eve wins $\Gt(\A)$ then she also wins $\Gt(\boxA)$.
\end{restatable}

The proof, given in \cref{app:G2-to-G2box}, is very similar in spirit to the proof of \cref{lem:GFG-preservation}: we consider a positional winning strategy of Eve in an intermediate game, where she plays her token in a copy of $\A$, against Adam playing in two copies of $\boxA$.


As deciding $\Gt$ on alternating automata is also \ptime, proving \cref{con:G2-to-GFG} would also provide a \ptime algorithm for deciding the GFGness of APW. In~\Cref{sec:G2-coBuchi} we work towards this goal by proving that $\Gt$ characterises GFG for nondeterministic coB\"uchi automata.

While for now we fall short of deciding GFGness of APW in \ptime, our technical developments  suffice to decide GFGness for alternating weak automata (AWW) in \ptime.

\begin{corollary}
Deciding whether an AWW $\A$ is GFG is in \ptime.
\end{corollary}

\begin{proof}
Recall that if $\A$ is an AWW, then both $\boxA$ and $\boxAco$ are B\"uchi automata.

We can then show that Eve wins $\Gt(\A)$ and $\Gt(\bar\A)$ if and only if $\A$ is GFG. Indeed, from \cref{pro:alt-to-nd} if Eve wins $\Gt(A)$ and $\Gt(\bar \A)$, she wins $\Gt(\boxA)$ and $\Gt(\boxAco)$. From \cref{thm:G2-buchi}, $\boxA$ and $\boxAco$ are then GFG, and therefore so is the nondeterminism of $\A$ and $\bar \A$, i.e.,\ $\A$ is GFG. The other direction follows from \cref{lem:gfg-to-G2}.

We can then solve $\Gt(\A)$.
\end{proof}

This contrasts in particular with the \pspace-hardness from \cref{sec:Alternating-behaviour}, which holds even for weak automata, of deciding whether the nondeterminism of $\A$ is GFG.

%% file: 6_cobuchi.tex
\newcommand{\sigmamove}{\sigma_{\mathit{move}}}
\section{Deciding GFGness of NCW via Two-Token Games}
\label{sec:G2-coBuchi}

The following theorem constitutes a step towards proving \cref{con:G2-to-GFG}. It also provides a simplified  \ptime algorithm for deciding whether an NCW is GFG: it suffices to solve $\Gt$.

\begin{theorem}
\label{thm:G2cobuchi}
A~nondeterministic~coB\"uchi~automaton~$\A$~is~GFG~if~and~only~if~Eve~wins~$G_2(\A)$.
\end{theorem}

We give only a proof sketch conveying the main ideas of the construction, leaving the detailed proof to \cref{ap:G2-coBuchi}.
The proof is inspired both by the construction from~\cite{BK18} for B\"uchi automata, and by techniques tailored to coB\"uchi automata from~\cite{KS15}.

It is straightforward that if an NCW $\A$ is GFG then Eve wins $G_2(\A)$ \cite{BK18}.  We thus assume that Eve wins $G_2(\A)$ and show that $\A$ is GFG.
\Subject{Global proof scheme}
Let us start by recalling the main proof scheme of \cite{BK18}, showing that for all NBW $\A$, if Eve wins $G_2(\A)$ then $\A$ is GFG:
\begin{itemize}
\item For every $k\in\Nat\setminus\{0,1\}$, Eve wins $G_2(\A)$ if and only if Eve wins $G_k(\A)$, a game where she has one token and Adam has $k$ tokens.
\item We assume, towards contradiction, that the automaton is not GFG, and we fix a finite\=/memory strategy $\stratA$ for Adam in the letter game of $\A$. This strategy chooses letters such that the produced word $w$ is always in $L(\A)$. Moreover, the finite memory of $\stratA$ guarantees additional structure on the run-DAG of $\A$ on $w$.
\item We describe a strategy $\sigmamove$ to move a fixed number $N$ of tokens in $\A$, such that any word produced by $\stratA$ will be accepted by one of the $N$ tokens.
\item Finally, we build a strategy $\stratE$ for Eve in the letter game of $\A$, moving $N$ virtual tokens in her memory according to $\sigmamove$, and playing her winning strategy $\sigma_N$ in $G_N(\A)$ against them. The play yielded by $\stratE$ playing against $\stratA$ will be winning for Eve, contradicting the fact that $\stratA$ is a winning strategy in the letter game.
\end{itemize}

\Subject{Switching to the coB\"uchi condition}
The goal is to use the same proof scheme. However the strategy $\sigmamove$ will be more involved. Indeed, for the B\"uchi condition, it was enough to take $\sigmamove$ to be a strategy that spread tokens evenly at each nondeterministic choice. This is no longer true for the coB\"uchi condition, and the main challenge here consists of building a strategy $\sigmamove$ with the same properties.
The following ideas are inspired by \cite{KS15}:
\begin{itemize}
\item We show that the automaton $\A$ can be taken in a form that guarantees properties related to the winning region of $G_2(\A)$ and the structure of the graph of accepting transitions.
\item We show that there is a subset $S$ of states of $\A$ and a deterministic transition function $\deltadet$ such that any word $w\in L(\A)$ is ultimately accepted from a state of $S$ via the run yielded by $\deltadet$, without any rejecting transition.
\end{itemize}

We also provide a new construction: we use the fact that Eve wins $G_k(\A)$ for every $k\in\Nat$ to build a ``limit strategy'' $\sigmainf$ in the letter game of $\A$. This strategy might build a rejecting run, but guarantees that for each state $p$ it reaches, and any number $k$ of tokens at reachable states $q_1\dots,q_k$, the position $(p;q_1,\dots ,q_k)$ is in the winning region of $G_k(\A)$.\\

We are now ready to build the strategy $\sigmamove$, which is the only missing piece to complete the proof.
We take a big number $N$ of tokens that depends on the size of $\A$ and of the size of the memory of Adam's strategy $\stratA$.
The strategy $\sigmamove$ moves these $N$ main tokens according to $\sigmainf$ until, one by one, they become \emph{active} and deviate from $\sigmainf$ to attempt to build an accepting run. To do so, the current active token will play $\sigma_{|S|}$, the winning strategy in $G_{|S|}$, against $|S|$ virtual \emph{deterministic tokens}. These $|S|$ tokens start from the states of $S$ that Adam could have reached, and move deterministically according to $\deltadet$. We use here the fact that $\sigmainf$ was built so that the active token is in a position to win $G_{|S|}(\A)$. If a deterministic token encounters a rejecting transition, it is considered \emph{dead}. If all deterministic tokens are dead, we reached a \emph{breakpoint}: a new main token becomes active, and deviates from $\sigmainf$ by starting to play $\sigma_{|S|}$ against $|S|$ new virtual deterministic tokens. On the other hand, if at least one of these deterministic tokens stays alive forever, then our currently active token will build an accepting run by correctness of $\sigma_{|S|}$. We show that this must happen eventually, as otherwise there are $N$ ``breakpoints'', contradicting the fact that the finite-memory strategy $\stratA$ only builds words in $L(\A)$. This means that the global strategy $\sigmamove$ is correct: one of the $N$ main tokens will always build an accepting run, providing the input word has been produced by $\stratA$.
The behaviour of tokens is illustrated in \cref{fig:tokens}.

\begin{figure}
\centering
\includegraphics[scale=.5]{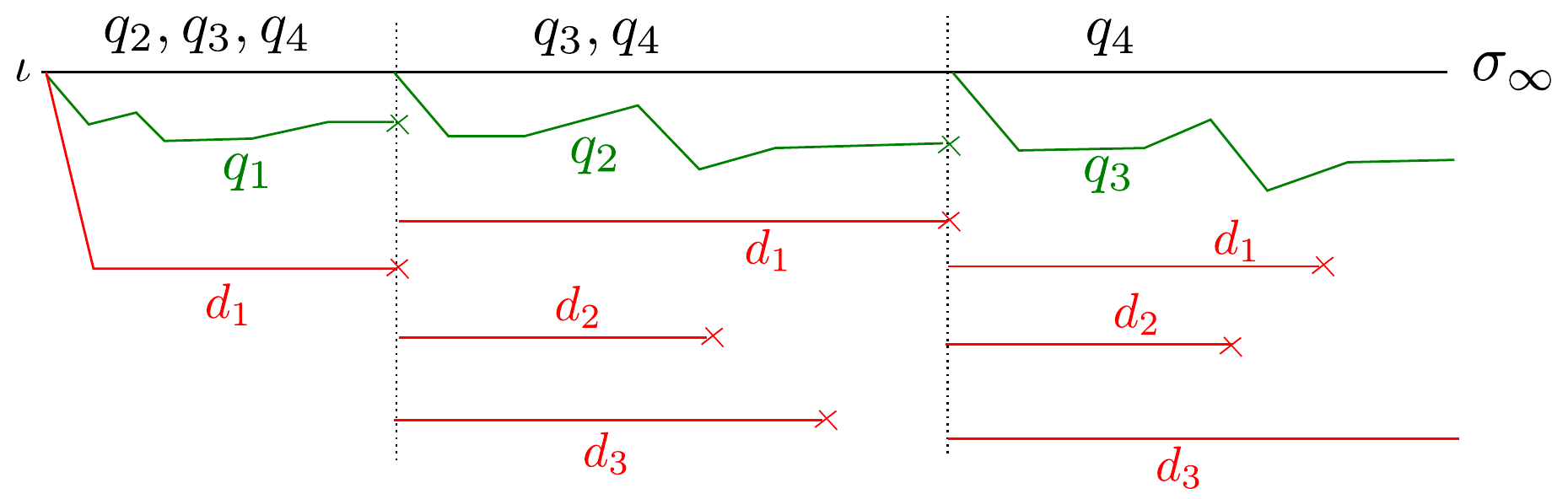}
\caption{An illustration of the behaviour of tokens in the memory structure of the strategy $\stratE$. Awaiting main tokens are represented in black, active main tokens in green, and alive deterministic tokens in red. Breakpoints are represented by dashed vertical lines.}
\label{fig:tokens}
\end{figure}

Finally, to win the letter game against $\tau$, Eve moves her token by simulating her strategy in $G_N(\A)$ against $N$ virtual tokens moving according to $\sigmamove$. Then, as $\tau$ plays a word in $L(\A)$, one of the $N$ virtual tokens is guaranteed to follow an accepting run, so Eve's token will also follow an accepting run. This contradicts the assumption that $\tau$ is a winning strategy for Adam, and proves that $G_2$ indeed characterises GFGness of NCW.

%% file: 7_concl.tex
\section{Conclusions}
\label{sec:conclusions}

The results obtained in this work shed new light on where alternating GFG automata resemble nondeterministic ones, and where they differ.

In particular, we show that alternating parity GFG automata can be exponentially more succinct than any equivalent nondeterministic GFG automata, yet this succinctness does not become double exponential when compared to deterministic automata, answering a question from~\cite{BL19}. Some further succinctness problems are left open here, such as the possibility of a doubly exponential gap between GFG automata of stronger acceptance conditions and deterministic ones, as well as between $\EGFG$ parity automata and deterministic ones.

We also show that the interplay between the two players can be used to decide whether an automaton is GFG without deciding $\EGFG$ and $\AGFG$ separately.
In particular, the $\Gt$ characterisation of nondeterministic GFG automata conjectured in~\cite{BK18} (see \cref{con:G2-to-GFG}) suffices to recognise alternating GFG parity automata of fixed index in \ptime.
We provide further evidence to the conjecture, by proving it for the coB\"uchi condition, combining insights from~\cite{KS15} and~\cite{BK18}, and using some new techniques. We also note that the conjecture holds for the generalized-B\"uchi condition, via a proof that slightly adapts the one from~\cite{BK18} for the B\"uchi  case.
Still, we could not manage to prove the conjecture for nondeterministic automata with $3$ parity priorities, and we believe that new insights will be necessary to climb up the parity ladder.

%% file: AP_2_prelim.tex
\section{Appendix of \cref{sec:Preliminaries}}
\label{ap:Preliminaries}

In this section of the appendix we provide the remaining technical definitions from \cref{sec:Preliminaries} that are used in the proofs.

\newcommand{\GL}{C} 
\Subject{Games}
A~\emph{$\Sigma$-arena} is a directed (finite or infinite) graph with nodes (positions) split into $E$\=/labelled positions of Eve and $A$-labelled positions of Adam, where the edges (transitions) are labelled by elements of $\Sigma\sqcup\{\noLab\}$. The role of $\noLab$ is to mark edges that have no influence on the winner of a play, e.g., edges allowing players to resolve some Boolean formula.

We represent such an arena as $R=(V,X,V_E,V_A)$, where $V$ is its set of positions; $X\subseteq V\times \big(\Sigma\sqcup\{\noLab\}\big)\times V$ its transitions; $V_E\subseteq V$ the $E$\=/positions; and $V_A = V\setminus V_E$ the $A$\=/positions.

Notice that the definition allows more than one transition between a pair of positions (such transitions needs to have distinct labels). We will require that each infinite path contains infinitely many $\Sigma$\=/labelled transitions. An arena might be rooted at an initial position $v_\iota\in V$. We say that a position $v$ is \emph{terminal} if there is no outgoing transition from~$v$ (i.e.\ no element of $X$ of the form $(v,\letter,v')$). If we don't say that an arena is \emph{partial} then it is assumed that there are no terminal positions.

If $R$ is a (partial) $\Sigma$\=/arena and $V'\subseteq V$ is a set of positions, then $R\restr_{V'}$ is the \emph{sub-arena} of $R$ defined as the restriction of $R$ to the positions in $V'$, namely for $P\in\{E,A\}$, the $P$\=/positions of $R\restr_{V'}$ are $V'_P:= V_P\cap V'$, and its transitions are $X':= X\cap (V'\times (\Sigma\cup\{\epsilon\})\times V')$.
We say that two (partial) $\Sigma$-arenas 
$R=(V,X,V_E,V_A)$ and $R'=(V',X',V'_E,V'_A)$ are \emph{isomorphic} if there exists a bijection $i\colon V\to V'$ that preserves the membership in $V_P$/$V'_P$, for $P\in\{E,A\}$, and sets of transitions $X$/$X'$.

A~\emph{partial play} in $R$ is a path in $R$, i.e.,\ an element $\pi=v_0e_0v_1e_1\ldots$ of $V\cdot\big(X\cdot V)^\ast\cup\big(V\cdot X\big)^\omega$, where for every $i$ we have $e_i=(v_i,\letter_i, v_{i+1})$. Such a partial play is said to \emph{begin} in $v_0$. A~partial play is a \emph{play} if either it is infinite or the last position $v_i$ is terminal. 

A \emph{game} is a $\Sigma$-arena together with a winning condition $W\subseteq \Sigma^\omega$. An infinite play $\pi$ is said to be winning for Eve in the game if the sequence of $\Sigma$-labels $(\letter_i)_{i\in\Nat}$ of the transitions along $\pi$ form a word in $W$. Else $\pi$ is winning for Adam. Games with some class $X$ of winning conditions (e.g.,\ the parity condition) are called $X$ games (e.g.,\ parity games).

A \emph{strategy} for Eve (resp.\ Adam) is a function $\tau\colon V\cdot \big(X\cdot V\big)^*\rightarrow X$ that maps a \emph{history} $v_0e_0v_1\ldots e_{i-1}v_i$, i.e. a finite prefix of a play in $R$, to a transition $e_i$ whenever $v_i$ belongs to $V_E$ (resp. to $V_A$). A partial play $v_0e_0v_1e_1\dots$ agrees with a strategy $\tau$ for Eve (Adam) if whenever $v_i\in V_E$ (resp. in $V_A$), we have $e_i=\tau(v_0e_0v_1\ldots e_{i-1}v_i)$. A strategy for Eve (Adam) is winning from a position $v\in V$ if all plays beginning in $v$ that agree with it are winning for Eve (Adam). We say that a player wins the game from a position $v\in V$ if they have a winning strategy from $v$.
If the game is rooted at $v_\iota$, we say that a player wins the game if they win from $v_\iota$.

A strategy is \emph{positional} if its value depends only on the last position, i.e.,\ $\tau(v_0e_0\cdots e_{i-1}v_i)$ depends only on $v_i$. In that case the strategy of a player $P$ can be represented as a function $\tau\colon V_P\to X$.

We also define the notion of \emph{strategy with memory $M$} for player $P$. This is a tuple $(\sigma, M, m_0, \upd)$ where $M$ is a~set of \emph{memory states}; $m_0$ is an~\emph{initial memory state}; $\upd\colon M\times X\to M$ is an \emph{update function}, and $\sigma\colon M\times V_P\to X$ is a~strategy deciding which move should be played, depending only on the current memory state and on the current position. Along a play, the memory starts with $m_0$, and is updated along every transition according to $\upd$. The general notion of strategy corresponds to $M=(V\cdot X)^*$, and positional strategies correspond to $M$ being a singleton.
A player has a \emph{finite\=/memory winning strategy} if there exists a winning strategy using a finite memory set $M$.

\begin{proposition}
\label{lem:unfolding}
Let $G$ and $G'$ be two $\Sigma$-games with the same winning condition, such that the unfoldings of $G$ and $G'$ are isomorphic. Then Eve has a winning strategy in $G$ if and only if she has a winning strategy in $G'$.
\end{proposition}

\begin{proposition}[\cite{klarlund_progress_positional}]
\label{prop:RabinPositionalDeterminacy}
Rabin games are positionally determined for Eve. (If Eve has a~winning strategy then she has a positional winning strategy.)
\end{proposition}

\begin{definition}[Synchronised product]
\label{def:GameAutomataProd}
The \emph{synchronised product} $\pro{R}{\A}$ of a (partial) $\Sigma$-arena $R=(V,X,V_E,V_A)$ and an alternating automaton $\A=(\Sigma,Q,\iota,\delta,\alpha)$ with a set of transitions $T$ and labelling $\alpha: T \to \Gamma$ is a (partial) $\Sigma{\times}\Gamma$-arena defined as follows. Its set of positions is $(V\times Q) \cup (V\times Q\times\Sigma\times\widehat{\delta})$, and its transitions are defined by:

\begin{tabular}{ll}
$\big\langle (v,q), \noLab, (v',q,\letter,\delta(q,\letter))\big\rangle$
&for $(v,q)\in V\times Q$ and $\langle v, \letter, v'\rangle\in X$;\\
$\big\langle (v,q), \noLab, (v',q)\big\rangle$ 
& for $(v,q)\in V\times Q$ and $\langle v, \noLab, v'\rangle\in X$;\\
$\big\langle (v,q,\letter,b), \noLab, (v,q,\letter,b_i)\big\rangle$
&for $(v,q,\letter,b)\in V\times Q\times \Sigma\times \widehat{\delta}$\\
&with $b=b_1{\lor} b_2$ or $b=b_1{\land} b_2$ and $i=1,2$;\\
$\big\langle (v,q,\letter,q'), \big(\letter,\alpha(q,\letter,q')\big), (v,q')\big\rangle$
&for $(v,q,\letter,q')\in V\times Q\times \Sigma\times \widehat{\delta}$ with $q'\in Q$.
\end{tabular}

The positions belonging to Eve are of the form $(v,q)$ where $v\in V_E$ and of the form $(v,q,\letter,b_1{\lor}b_2)$. The remaining ones belong to Adam. If $R$ has an initial position $v_\iota$ then the initial position of the product is $(v_\iota,\iota)$.
\end{definition}

We implicitly assume that the arena only contains vertices that are reachable from $V\times Q$. (They need not be reachable from an initial position of $R$ and an initial state of $\A$, but from some position of $R$ and state of $\A$.)

We will sometimes consider longer products, like $(R\times \A_1)\times \A_2$, where $R$ is a~$\Sigma$\=/arena and both automata $\A_1$ and $\A_2$ are over the alphabet $\Sigma$. Assume that $\A_1$ and $\A_2$ have transitions labelled in sets $\Gamma_1$ and $\Gamma_2$ respectively. Notice that in that case the arena $R\times \A_1$ is formally a~$\Sigma\times\Gamma_1$\=/arena. Thus, to make the above formula precise, we treat the automaton $\A_2$ as an~automaton over the alphabet $\Sigma\times\Gamma_1$ and assume that it ignores the second component of the letters read.

\Subject{One-step arenas}
For a~letter $\letter\in\Sigma$, we denote by $R_\letter$ a~partial $\Sigma$\=/arena consisting of two vertices $v$ and $v'$ (it does not matter which of the players controls them), and one transition $\tuple{v, \letter, v'}$. Then, for an automaton $\A=(\Sigma,Q,\iota,\delta,\alpha)$, the product $R_\letter\times \A$ is a partial arena, in which the players should resolve their choices in the formulas $\delta(q,\letter)$ for all the possible states $q\in Q$. We call it the \emph{one\=/step arena} of $\A$ over $\letter$.
Such an arena contains one position of the form $(v,q)$ for each state $q\in Q$; a~set of non\=/terminal positions of the form $(v',q,\letter,\psi)$ for some $q\in Q$ and $\psi\in\widehat{\delta}$; and one terminal position of the form $(v',q)$ for each state $q\in Q$. (See \cref{fig:one-step-arena}.)

\begin{figure}
\centering
\input{Fig_1_on-step}
\caption{A one-step arena over a letter $a\in\Sigma$, obtained as a product of a~simple arena $R_a$ with the alternating parity automaton $\A$. In this example $v$ is controlled by Eve and $v'$ by Adam. The transitions with no label are labelled by $\noLab$. Diamond\=/shaped positions belong to Eve and square\=/shaped positions belong to Adam.}
\label{fig:one-step-arena}
\end{figure}
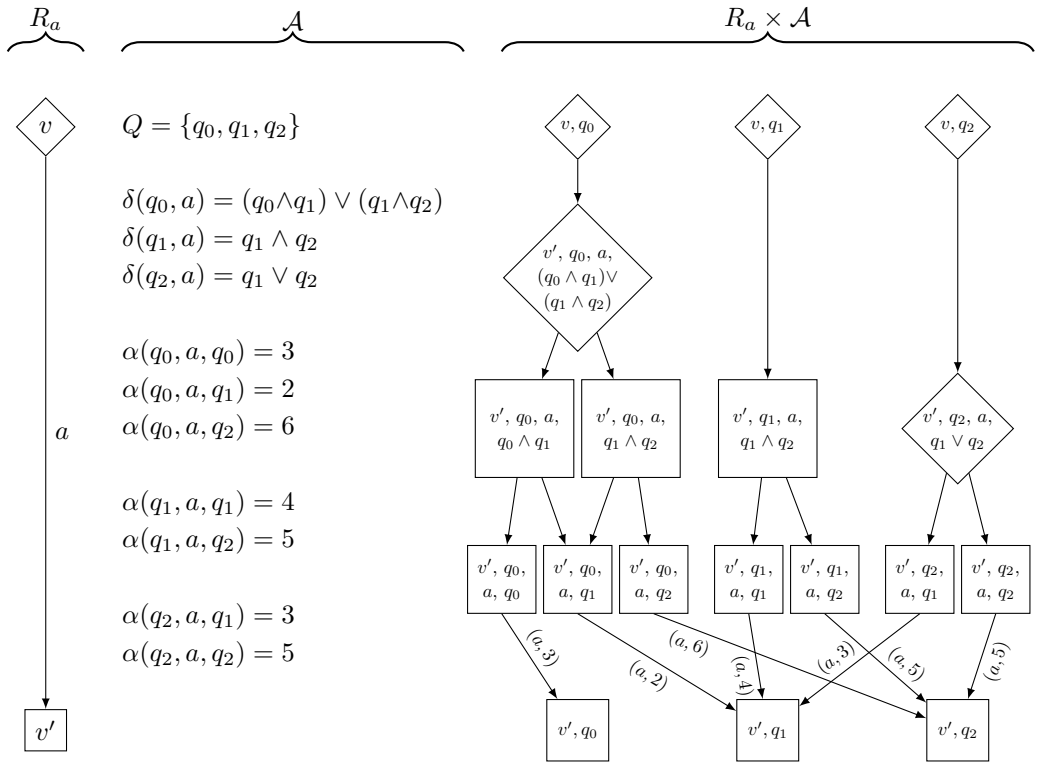

\Subject{Boxes}
\label{def:boxes}
In the later exposition, we will be interested in the combinatorial structure of possible strategies of Eve over one-step arenas $R_\letter\times \A$. For a positional strategy $\strat$ of Eve in a game of the form $R_\letter\times \A$, we define the \emph{box} of $\A$, $\letter$, and $\strat$, denoted by $\beta(\A,\letter,\strat)$, as the relation that is a subset of $Q\times \Sigma\times Q$ and contains a triple $(q,\letter,q')$ iff there exists a play in $R_\letter\times \A$ that is consistent with $\strat$, starting in $(v,q)$ and ending in $(v',q')$. We further define for every $\letter\in\Sigma$, the set $\boxes_{\tuple{\A,\letter}} = \{\beta(\A,\letter,\strat) \mid\text{$\strat$ is a positional strategy of Eve}\}$. Finally, let $\boxes_\A:=\bigcup_{\letter\in \Sigma}\boxes_{\tuple{\A,\letter}}$. Notice that $|\boxes_\A| \leq 2^{|Q\times \Sigma\times Q|}$.
When speaking of an arbitrary box, we mean any non-empty relation $\beta\subseteq Q\times\Sigma\times Q$ where all the letters $\letter$ appearing on the middle component are equal.

\cref{fig:boxes} represents $\boxes_{\tuple{\A,a}}$ for the automaton $\A$ of \cref{fig:one-step-arena}: Since there are two binary\=/choice positions of Eve in the corresponding one\=/step arena, there are four distinct positional strategies of Eve, which give the four possible boxes. They correspond to Eve choosing respectively LL, LR, RL, RR, where L stands for a left choice and R for a right choice in each of her two binary\=/choice positions.

\begin{proposition}
\label{prop:choice-to-strat}
Consider a letter $\letter$ and an automaton $\A$ with states $Q$ and transition function $\delta$. Then there is a bijection between $\boxes_{\tuple{\A,\letter}}$ and the positional strategies of Eve in the one-step arena of $\A$ and $\letter$.
\end{proposition}

\begin{figure}
\centering
\input{Fig_2_boxes}
\caption{The four possible boxes corresponding to Eve's choices in the one-step arena of \cref{fig:one-step-arena}. (All edges should be labelled with $a$, which we omit for better readability.)}
\label{fig:boxes}
\end{figure}
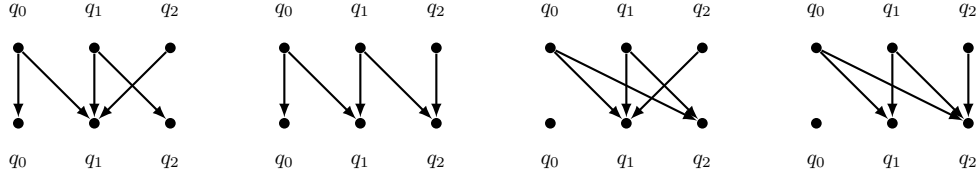

\Subject{Acceptance of a word by an automaton}
We define the acceptance directly in terms of the model-checking (acceptance/membership) game, which happens to be exactly the product of the automaton with a path-like arena describing the input word. More precisely, given a word $w\in\Sigma^\omega$, the \emph{model-checking game} is defined as the product $\pro{R_w}{\A}$, where the arena~$R_w$ consists of an infinite path $\omega$, of which all positions belong to Eve (although it does not matter); the transitions are of the form $\langle i, w_i, i{+}1\rangle$; the initial position is $0$; and the winning condition is based on the winning condition of $\A$ (the $\Sigma$-component of the labels is ignored). We say that $\A$ \emph{accepts} $w$ if Eve has a winning strategy in the model-checking game $R_w\times \A$. The language of an automaton $\A$, denoted by $L(\A)$, is the set of words that it accepts (recognises).

Notice that for each $i\in\Nat$, the sub-arena of $\pro{R_w}{\A}$ with positions in
\[\big( \{i\}\times Q\big) \cup \big(\{i{+}1\} \times  Q{\times}\Sigma{\times}\BP(Q)\big)\cup \big(\{i{+}1\}\times Q\big)\]
is isomorphic to the one-step arena of $\A$ over $w_i$.

\NotNeeded{
\begin{figure}
\centering
\input{Fig_3_model-check}
\caption{The model-checking game defined as the product of the arena $R_w$ representing a given word $w$ and the automaton $\A$ of \cref{fig:one-step-arena}. The one-step arena of \cref{fig:one-step-arena} depicts the steps that correspond to the letter $a$.}
\label{fig:model-cheking-game}
\end{figure}
}

%% file: Fig_1_on-step.tex
\begin{tikzpicture}

\draw (-0.5, 1) edge[obrace] node{$R_a$} (0.5, 1);
\draw (1, 1) edge[obrace] node{$\A$} (5.5, 1);
\draw (6, 1) edge[obrace] node{$R_a\times \A$} (13, 1);

\node[eve,  minimum size=22pt] (vi) at (0, 0) {$v$};
\node[adam, minimum size=22pt] (vt) at (0,-8) {$v'$};
\draw[trans] (vi) edge node {$a$} (vt);

\tikzstyle{flow} = [anchor=mid west, inner sep=0cm]
\newcommand{\x}{1}

\node[flow] at (\x, -0.0) {$Q=\{q_0, q_1, q_2\}$};

\node[flow] at (\x, -1.0) {$\delta(q_0,a)=(q_0{\land} q_1)\lor(q_1{\land}q_2)$};
\node[flow] at (\x, -1.5) {$\delta(q_1,a)=q_1 \land q_2$};
\node[flow] at (\x, -2.0) {$\delta(q_2,a)=q_1 \lor q_2$};

\node[flow] at (\x, -3.0) {$\alpha(q_0,a,q_0)=3$};
\node[flow] at (\x, -3.5) {$\alpha(q_0,a,q_1)=2$};
\node[flow] at (\x, -4.0) {$\alpha(q_0,a,q_2)=6$};

\node[flow] at (\x, -5.0) {$\alpha(q_1,a,q_1)=4$};
\node[flow] at (\x, -5.5) {$\alpha(q_1,a,q_2)=5$};

\node[flow] at (\x, -6.5) {$\alpha(q_2,a,q_1)=3$};
\node[flow] at (\x, -7.0) {$\alpha(q_2,a,q_2)=5$};


\foreach \x/\p in {0/eve, 1/adam, 2/adam} {
	\node[eve,  scale=0.7, minimum size=35pt] (v\x) at (7 + 2.5*\x, +0) {$v, q_{\x}$};
	\node[adam, scale=0.7, minimum size=30pt] (f\x) at (7 + 2.5*\x, -8) {$v', q_{\x}$};
}

\tikzstyle{Atree} = [adam, scale=0.7, inner sep=0cm, node distance=0cm and 0cm, align=center]
\tikzstyle{Etwo} = [eve, scale=0.7, inner sep=0cm, node distance=0cm and 0cm, align=center]

\node[Etwo] (w0) at ($(v0)+(-0,-2)$) {$v'$, $q_0$, $a$, \\ $(q_0\land q_1)\lor$\\$(q_1\land q_2)$};
	
\tikzstyle{Atree} = [adam, scale=0.7, inner sep=0cm, node distance=0cm and 0cm, align=center]
\tikzstyle{Etree} = [eve, scale=0.7, inner sep=0cm, node distance=0cm and 0cm, align=center]

\node[Atree] (wl0) at ($(v0)+(-0.7,-4)$) {$v'$, $q_0$, $a$,\\$q_0\land q_1$};

\node[Atree] (wr0) at ($(v0)+(+0.7,-4)$) {$v'$, $q_0$, $a$,\\$q_1\land q_2$};

\node[Atree] (wl1) at ($(v1)+(+0.0,-4)$) {$v'$, $q_1$, $a$,\\$q_1\land q_2$};

\node[Etree] (wr2) at ($(v2)+(-0,-4)$) {$v'$, $q_2$, $a$,\\$q_1\lor q_2$};

\tikzstyle{Afour} = [adam, scale=0.7, inner sep=0cm, node distance=0cm and 0cm, align=center]
\tikzstyle{Efour} = [eve, scale=0.7, inner sep=0cm, node distance=0cm and 0cm, align=center]

\node[Afour] (f00) at ($(v0)+(-1,-6)$) {$v'$, $q_0$,\\$a$, $q_0$};
\node[Afour] (f01) at ($(v0)+(+0,-6)$) {$v'$, $q_0$,\\$a$, $q_1$};
\node[Afour] (f02) at ($(v0)+(+1,-6)$) {$v'$, $q_0$,\\$a$, $q_2$};

\node[Afour] (f11) at ($(wl1)+(-0.25,-2)$) {$v'$, $q_1$,\\$a$, $q_1$};
\node[Afour] (f12) at ($(wl1)+(+0.75,-2)$) {$v'$, $q_1$,\\$a$, $q_2$};

\node[Afour] (f21) at ($(wr2)+(-0.5,-2)$) {$v'$, $q_2$,\\$a$, $q_1$};
\node[Afour] (f22) at ($(wr2)+(+0.5,-2)$) {$v'$, $q_2$,\\$a$, $q_2$};

\draw[trans] (v0) edge (w0);	
\draw[trans] (v1) edge (wl1);	
\draw[trans] (v2) edge (wr2);	

\draw[trans] (w0) edge (wl0);
\draw[trans] (w0) edge (wr0);

\draw[trans] (wl0) edge (f00);
\draw[trans] (wl0) edge (f01);
\draw[trans] (wr0) edge (f01);
\draw[trans] (wr0) edge (f02);

\draw[trans] (wl1) edge (f11);
\draw[trans] (wl1) edge (f12);

\draw[trans] (wr2) edge (f21);
\draw[trans] (wr2) edge (f22);

\tikzstyle{edtr} = [sloped, pos=0.5, anchor=south, scale=0.7]

\draw (f00.south) edge[trans] node[edtr, above] {$(a,3)$} (f0);
\draw (f01.south) edge[trans] node[edtr, below] {$(a,2)$} (f1);
\draw (f02.south) edge[trans] node[edtr, pos=0.15, below] {$(a,6)$} (f2);

\draw (f11.south) edge[trans] node[edtr, pos=0.7, below] {$(a,4)$} (f1);
\draw (f12.south) edge[trans] node[edtr, pos=0.7, above] {$(a,5)$} (f2);

\draw (f21.south) edge[trans] node[edtr, pos=0.61, above] {$(a,3)$} (f1);
\draw (f22.south) edge[trans] node[edtr, below] {$(a,5)$} (f2);
\end{tikzpicture}

%% file: Fig_2_boxes.tex
\newcommand{\drawBox}[2]{
\foreach \x in {0,1,2} {
  \node[scale=0.8] at (#1+\x, #2+0.5) {$q_\x$};
  \node[dot] (u\x) at (#1+\x, #2+0) {};
  \node[dot] (l\x) at (#1+\x, #2-1) {};
  \node[scale=0.8] at (#1+\x, #2-1.5) {$q_\x$};
}
}

\begin{tikzpicture}
\drawBox{0}{0}
\draw[edgeBox] (u0) -- (l0);
\draw[edgeBox] (u0) -- (l1);

\draw[edgeBox] (u1) -- (l1);
\draw[edgeBox] (u1) -- (l2);

\draw[edgeBox] (u2) -- (l1);

\drawBox{3.5}{0}
\draw[edgeBox] (u0) -- (l0);
\draw[edgeBox] (u0) -- (l1);

\draw[edgeBox] (u1) -- (l1);
\draw[edgeBox] (u1) -- (l2);

\draw[edgeBox] (u2) -- (l2);

\drawBox{7}{0}
\draw[edgeBox] (u0) -- (l1);
\draw[edgeBox] (u0) -- (l2);

\draw[edgeBox] (u1) -- (l1);
\draw[edgeBox] (u1) -- (l2);

\draw[edgeBox] (u2) -- (l1);

\drawBox{10.5}{0}
\draw[edgeBox] (u0) -- (l1);
\draw[edgeBox] (u0) -- (l2);

\draw[edgeBox] (u1) -- (l1);
\draw[edgeBox] (u1) -- (l2);

\draw[edgeBox] (u2) -- (l2);
\end{tikzpicture}

%% file: Fig_3_model-check.tex
\begin{tikzpicture}
\draw (-1, 1) edge[obrace] node{$R_w$} (1, 1);

\draw (3, 1) edge[obrace] node{$R_w\times \A$} (9, 1);

\foreach \x in {0, 1, 2, 3} {
	\node[eve, scale=0.9] (l\x) at (0, -2* \x) {$\x$};
	
	\foreach \s in {0,...,2} {
		\node[eve, scale=0.7] (v\x\s) at (4 + 2*\s, -2 * \x) {$(\x, q_{\s})$};
	}
}

\foreach \x/\l in {0/a, 1/b, 2/b} {
	\evalInt{\y}{\x+1}
	\draw[trans] (l\x) edge node[scale=0.8] {$\l$} (l\y);
	\node[inner sep=1pt] at (6, -2*\x - 1) {{\it one-step arena over \l}};
}

\node[dots] at (0, -7) {$\vdots$};
\node[dots] at (4, -7) {$\vdots$};
\node[dots] at (8, -7) {$\vdots$};

\end{tikzpicture}

%% file: AP_3_alt-gfg.tex
\section{Appendix of \cref{sec:Alternating-behaviour}}
\label{ap:Alternating}

\begin{definition}[A formalisation of \cref{def:LetterGames}]
Let $R_{A,\Sigma}$ be the $\Sigma$-arena consisting of a single position $v$ that belongs to Adam and the set of transitions $X$ of the form $\tuple{v,\letter,v}$ for each letter $\letter\in \Sigma$ (see \cref{fig:letter-giving-game}). The arena $R_{E,\Sigma}$ is the same except that $v$ belongs to Eve. Notice that the products $R_{A,\Sigma}\times \A$ and $R_{E,\Sigma}\times \A$ are both labelled by $\Sigma\times \Gamma$, where $\Sigma$ is the alphabet of $\A$ and $\Gamma$ is $\A$'s labelling, on top of which its acceptance condition is defined. Thus, the winning condition of games defined on these arenas can depend on a sequence of labels of the form $(\letter_i,\gamma_i)_{i\in\Nat}$. Then, \emph{Eve's letter game} is played over  $R_{A,\Sigma}\times \A$, where Eve wins if:
\[\text{$(\letter_i)_{i\in\Nat}\notin L(\A)$ or the sequence $(\gamma_i)_{i\in\Nat}$ satisfies the acceptance condition of $\A$.}\]
Dually, \emph{Adam's letter game} is played over $R_{E,\Sigma}\times \A$, where Adam wins if:
\[\text{$(\letter_i)_{i\in\Nat}\in L(\A)$ or the sequence $(\gamma_i)_{i\in\Nat}$ violates the acceptance condition of $\A$.}\]
\end{definition}

\subsection{Proof of \cref{lem:Cn_family}}\label{ap:Cn}

From \cite{KS15}, there is a family $(\A_n)_{n\in\Nat}$ of GFG-NCWs with $n$ states over a fixed alphabet $\Sigma$, such that every DPW for $L_n = L(\A_n)$ is of size $2^{\Omega(n)}$. 
For every $n\in\Nat$, let $\B_n$ be the dual of $\A_n$, so $\B_n$ is a UBW accepting $\overline{L_n}$. 
We build an APW $\C_n$ over $\Sigma$ of size linear in $n$, 
by setting its initial state to move to the initial state of $\A_n$ when reading the letter $a\in\Sigma$ and to the initial state of $\B_n$ when reading the letter $b\in\Sigma$. 
The acceptance condition of $\C_n$ is a parity condition with priorities $\{0,1,2\}$: accepting transitions of $\A_n$ are assigned priority $0$, and accepting transitions of $\B_n$ priority $2$. Other transitions have priority $1$.

The automaton $C_n$ is represented below:

\begin{center}
\includegraphics[scale=.5]{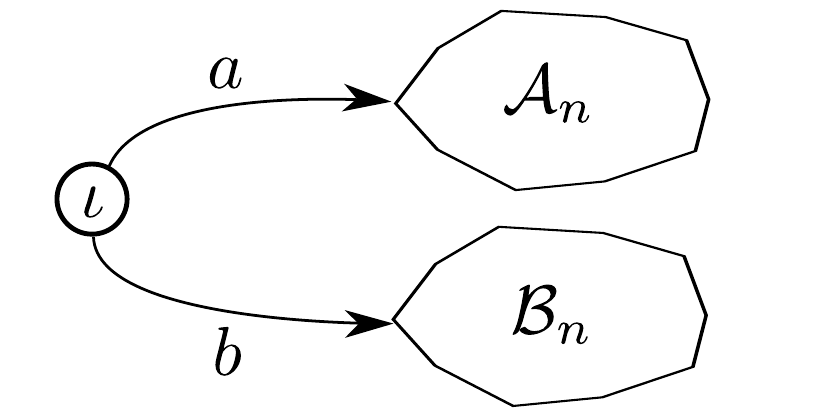}
\end{center}

Observe that $L(\C_n) = a L_n \cup b \overline{L_n}$, and that $\C_n$ is GFG: its initial state has only deterministic transitions, and over the $\A_n$ and $\B_n$ components, the strategy to resolve the nondeterminism and universality, respectively, follows the strategy to resolve the nondeterminism of $\A_n$, which is guaranteed due to $\A_n$'s GFGness.

Consider a GFG UPW $\E_n$ for $L(\C_n)$, and let $q$ be a state to which $\E_n$ moves when reading $a$, according to some strategy that witnesses $\E_n$'s GFGness. Then $\E_n^q$ is a GFG UPW for $L_n$.  Its dual is therefore a GFG NPW $\E'_n$ for $\overline{L_n}$.

Since $\A_n$ is a GFG NPW for $L_n$, by \cite[Thm 4]{BKKS13} we obtain a DPW for $L_n$ of size $|\A_n| |\E'_n|$. By choice of $L_n$, this DPW must be of size $2^{\Omega(n)}$, and since $\A_n$ is of size $n$, it follows that $\E'_n$, and hence $\E_n$, must be of size $2^{\Omega(n)}$. By a symmetric argument, every GFG NPW for $L(\C_n)$ must also be of size $2^{\Omega(n)}$.

\subsection{Proof of \cref{lem:EGFG_PSPACE}}\label{ap:pspace}

 Let $\A$ be an NFA over an alphabet $\Sigma=\{a,b\}$ and $\bar\A$ its dual. We want to check whether $L(\A)=\Sigma^*$.
We build an AFA $\B$, as depicted below, by first making Eve guess the second letter. If her guess is wrong, the automaton proceeds to a rejecting sink state $\bot$. Otherwise, it proceeds to the initial state of $\bar{A}$. The size of $\B$ is linear in the size of $\A$.
\begin{center}
\includegraphics[scale=.5]{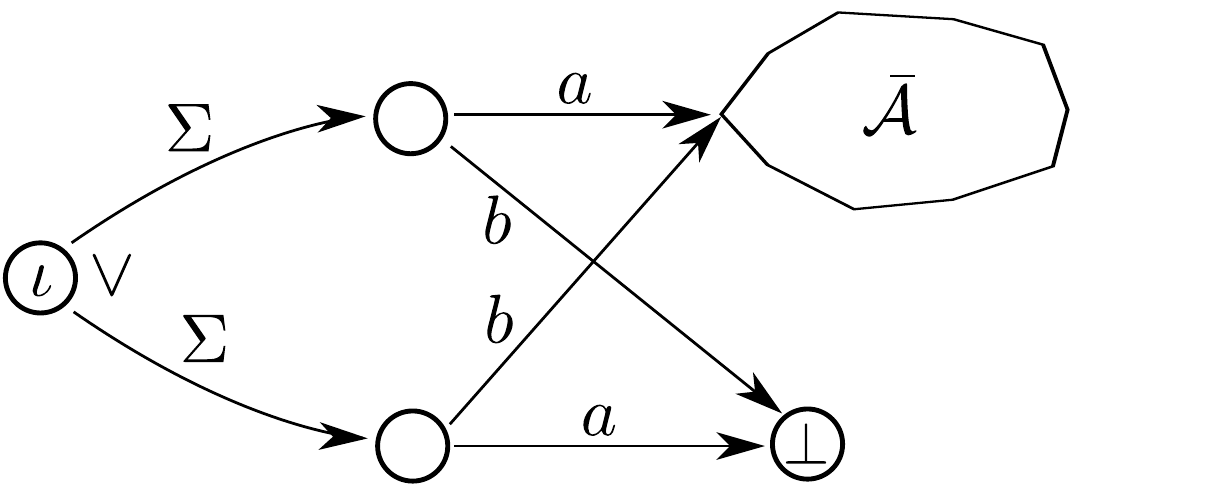}
\end{center}

If $(\bar{A})=\emptyset$, then $L(\B)=\emptyset$, so $\B$ is trivially $\EGFG$.
However, if there is some $u\in\bar{\A}$, then Adam has a winning strategy in Eve's letter game on $\B$. This strategy consists of playing $a$, then playing the letter that brings Eve to $\bot$, and finally playing $u$. The resulting word is in $L(\B)=\Sigma^2L(\A)$, so this witnesses that $\B$ is not $\EGFG$.
We obtain that $L(\A)=\Sigma^* \Leftrightarrow L(\bar{A})=\emptyset\Leftrightarrow\B\text{ is }\EGFG$, which is the wanted reduction.

%% file: AP_4_determ.tex
\section{Appendix of \cref{sec:Determinisation}}
\label{ap:Determinisation}

This section provides the technical details of the determinisation procedure in~\cref{sec:Determinisation}.
We start with some technical analysis of the types of histories needed to win letter games.

\subsection{Good for Games Automata: Required Histories}
\label{ap:sec:GfgDefinitions}

We begin by considering an~\emph{expanded} letter game. This will allow us to use a form of positional determinacy in letter games.

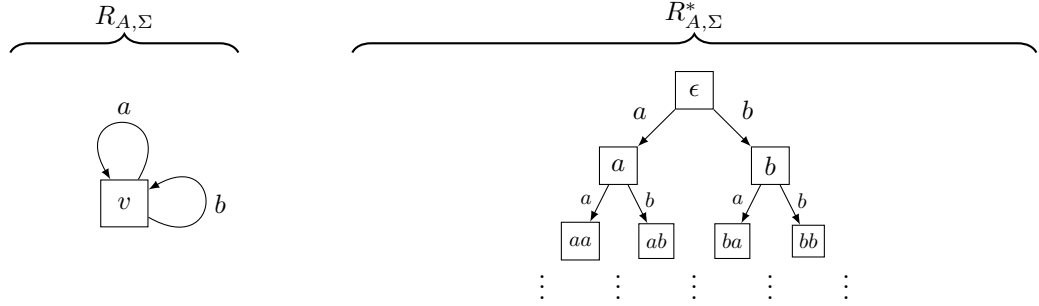
\begin{figure}[h]
\centering
\input{Fig_4_expanded}
\caption{The arenas $R_{A,\Sigma}$ and $R^\ast_{A,\Sigma}$, allowing Adam to choose an arbitrary sequence of letters. We define Eve's letter game for an automaton $\A$ over the product of $R_{A,\Sigma} \times \A$, and her expanded letter game over $R^\ast_{A,\Sigma} \times \A$; 
She wins a play if the word generated by Adam is not in $L(\A)$ or the path generated by her (resolving $\A$'s nondeterminism) and by Adam (resolving $\A$'s universality) satisfies $\A$'s acceptance condition.}
\label{fig:letter-giving-game}
\end{figure}

\Subject{Expanded letter games}
The definition of the letter games (\cref{def:LetterGames}) has the important advantage of being defined over a finite-arena. Yet, as a result, these games generally do not allow for positional determinacy. 

We provide below an expanded variant of the letter game that will have same unfolding as the original one, while being defined over an infinite arena. This will allow Eve to have positional determinacy in these games for Rabin automata.

Let $R^\ast_{A,\Sigma}$ be the $\Sigma$-arena with the set of positions $V=\Sigma^\ast$, all belonging to Adam, and the set of transitions $X$ of the form $\langle w, a, w\cdot \letter\rangle$ for each word $w\in\Sigma^\ast$ and letter $\letter\in\Sigma$ (see \cref{fig:letter-giving-game}). The initial position of this arena is $\epsilon$. The arena $R^\ast_{E,\Sigma}$ is the same, except that all the positions belong to Eve. We define \emph{Eve's expanded letter game} over $R^\ast_{A,\Sigma}\times \A$ and \emph{Adam's expanded letter game} over $R^\ast_{E,\Sigma}\times \A$ with the same winning conditions as in their (non expanded) variants.
The following follows directly from \cref{lem:unfolding}.

\begin{proposition}
For every automaton $\A$, the expanded letter games for $\A$ have the same winners as the (standard) letter games for $\A$.
\end{proposition}

\Subject{History Requirement}
Although  $R^\ast_{A,\Sigma}$ and $R^\ast_{E,\Sigma}$ are trees, 
the arenas of the expanded letter games are directed acyclic graphs as there can exist two distinct paths from the initial position to a~given position $(w,q)$. Thus, a priori, a winning strategy of a~player of such a game might need some history of a~play. However, as expressed by the following theorem, it is not the case.

\begin{theorem}\label{thm:ExpandedGamePositional}
If $\A$ is a Rabin (or parity) automaton then Eve's expanded letter game is positionally determined for Eve.
\end{theorem}

\begin{proof}
We will show that the winning condition of this game can be represented as a Rabin condition and invoke \cref{prop:RabinPositionalDeterminacy}. Let $\D'$ be a deterministic parity automaton recognising the complement of the language $L(\A)$, for a Rabin automaton $\A$ over the alphabet $\Sigma$. Let $\Omega\colon Q_{\D'}\to\Nat$ be the priority assignment of $\D'$ (without loss of generality we can assume that the states of $\D'$ bear priorities). 

Consider the arena $R'^\ast_{A,\Sigma}$ that is derived from the arena $R^\ast_{A,\Sigma}$ by adding to transitions priorities according to the deterministic of runs of $D'$, that is, by changing the labelling of every transition $\langle w, \letter, w\cdot \letter\rangle$ to $\langle w, (\Omega_{\D'}(q),\letter), w\cdot \letter\rangle$, where $\letter\in\Sigma$ and $q$ is the state of $\D'$ reached after reading the word $w$ from the initial state of $\D'$.

Consider the product $R'^\ast_{A,\Sigma}\times \A$, in which for $\A's$ transitions we ignore these additional labels. The labels of that product are now of the form $(\ell,\letter,\gamma)$, where $\ell\in\Nat$ is a priority of $\D'$, $\letter\in\Sigma$, and $\gamma\in\Gamma$ is a label of $\A$. Notice that when one forgets about the first coordinate of the label, the game is equal to $R^\ast_{A,\Sigma}\times \A$. Moreover, given a~sequence of labels $(\ell_i,\letter_i,\gamma_i)_{i\in\N}$, by the choice of $\D'$, we know that $(a_i)_{i\in\N}\notin L(\A)$ if and only if the sequence $(\ell_i)_{i\in\N}$ satisfies the parity condition.

Define the game $\G$ over $R'^\ast_{A,\Sigma}\times \A$, in which Eve wins a play labelled by $(\ell_i,\letter_i,\gamma_i)_{i\in\Nat}$ if
\[\text{$(\ell_i)_{i\in\Nat}$ satisfies the parity condition of $\D'$ or $(\gamma_i)_{i\in\Nat}$ satisfies the Rabin condition of $\A$.}\]
Notice that both disjuncts above can be written as Rabin conditions and therefore $\G$ is positionally determined for Eve. Moreover, the choice of $\D'$ guarantees that the new winning condition is equivalent to Eve's condition in her expanded letter game on $\A$---the same plays are winning for Eve in $\G$ and her expanded letter game on $\A$. Since the structure of the game is also preserved, it means that Eve's expanded letter game on $\A$ is positionally determined for Eve.
\end{proof}

\begin{remark}
Dually, Adam's expanded letter game for a Streett automaton is positionally determined for Adam.
\end{remark}

As a consequence  of \cref{thm:ExpandedGamePositional}, for alternating $\EGFG$ Rabin automata, a strategy for Eve to resolve the nondeterminism may ignore the history of the play, and only consider the history of the word read, as is the case for nondeterministic GFG automata.

We will now argue that Eve's positional strategy $\stratE$ in the expanded letter game on an alternating automaton $\A$ can be represented as a function $\stratE'\colon \Sigma^+\rightarrow \boxes_\A$ that assigns to each word $w\letter$ a box $\beta_{w\letter}\in\boxes_{\tuple{\A,\letter}}$. Indeed, let $(w\cdot\letter)\in\Sigma^+$ and let $V_{w\letter}$ be the set of positions of $R^\ast_{A,\Sigma}\times \A$ of the form $(w,q)$, $(w\cdot \letter, q,\letter,b)$, or $(w\cdot \letter, q)$ for $q\in Q$ and $b\in \BP(Q)$. Observe that the partial arena of $R^\ast_{A,\Sigma}\times \A$ restricted to $V_{w\letter}$ is isomorphic to the one-step arena $R_\letter\times \A$. Thus, $\stratE$ provides a positional strategy over this arena, which by \cref{prop:choice-to-strat} can be encoded as a box $\beta_{w\letter}$. More formally, let $\beta_{w\letter}$ contain $(q,a,q')$, if there is a~play consistent with $\stratE$ that visits both the positions $(w,q)$ and then $(w\letter,q')$.

Then, in the next lemma we show that if $\stratE$ is also winning, then the sequences of boxes $\beta_{w\letter}$ only has accepting paths.

\begin{definition}
\label{def:univ-acc-box}
Consider an automaton $\A$ with states $Q$ and initial state $\iota$, and an infinite word $u=\beta_0,\beta_1,\ldots\in(\boxes_\A)^\omega$. We say that a sequence of transitions $\rho=(q_i,\letter_i,q_{i+1})_{i\in\Nat}$ is a \emph{path of $u$} if $q_0=\iota$ and for every $i\in\Nat$, we have $(q_i,\letter_i q_{i+1}) \in \beta_i$. The word $u$ is \emph{universally accepting for $\A$} if each of its paths satisfies the acceptance condition of $\A$.
\end{definition}

\begin{lemma}
\label{lem:strat-to-win-boxes}
Given an alternating $\EGFG$ Rabin automaton $\A$, there is a positional strategy $\stratE$ in her expanded letter game on $\A$ such that
for every word $w\in L(\A)$ the sequence of boxes $u=\beta_0,\beta_1,\ldots\in(\boxes_\A)^\omega$ defined as $\beta_i=\stratE(w\restr_{i+1})$ is universally accepting for $\A$.
\end{lemma}

\begin{proof}
Consider words $w$ and $u$ as above. Let $\rho=(q_i,\letter_i,q_{i+1})_{i\in\Nat}$ be a path of $u$. Since the strategy $\stratE$ is positional, the definition of $\stratE(w\restr_{i+1})$ implies that there exists a~single play of the expanded letter game that visits all the positions of the form $(w\restr_i,q_i)$ for $i\in\N$. Since $w\in L(\A)$, the winning condition of the expanded letter game guaranteees that the path $\rho$ must be accepting.
\end{proof}

Observe that the above arguments do not hold for alternating GFG Streett automata: Since Streett games are not positionally determined for Eve, Eve's expanded letter game for a Streett automaton is not positionally determined for Eve (an analogous of \cref{thm:ExpandedGamePositional} does not hold). Furthermore, we provide in \cref{fig:StreetGFGHistory} an example of an alternating GFG Streett automaton, in which Eve cannot resolve her nondeterminism only according to the history of the word read.

\begin{proposition}
Consider an $\EGFG$ alternating Streett automaton $\A$ with transition conditions in DNF. Then Eve might not have a strategy~$\stratE\colon \Sigma^+\rightarrow \boxes_\A$ satisfying Lemma~\ref{lem:strat-to-win-boxes}.
\end{proposition}

\begin{proof}
Consider the ASW depicted in \cref{fig:StreetGFGHistory}. It is $\EGFG$, as witnessed by the strategy that chooses the transition $t_4$ in $q_3$ if the last visited state was $q_1$ and $t_5$ otherwise. Yet, there is no strategy that only remembers the word read so far, as this only gives the length of the word, and cannot help in determining whether the path visited $q_1$ or $q_2$.
\end{proof}

\begin{figure}
\centering
\input{Fig_5_history}
\caption{A GFG ASW over  a singleton alphabet, for which Eve's $\EGFG$ strategy cannot only remember the prefix of the word read so far, but also some history about the visited states.}
\label{fig:StreetGFGHistory}
\end{figure}

Interestingly, the question of whether Eve can resolve the nondeterminism in a class of alternating GFG automata with only the knowledge of the word read so far does not tightly correspond to whether the acceptance condition of this class is memoryless. For example, it does hold for the generalised-B\"uchi condition, though it is not memoryless.


\subsection{Alternation Removal in GFG Rabin Automata}
\label{app:alt-rem-gfg-rabin}

This section presents the proof of the following theorem:

\thmexpgfgdealt*


\lempropofd*

\begin{proof}
Notice that it is easy to construct a nondeterministic Streett (resp.\ parity) automaton $\S$ over the alphabet $\boxes_\A$ that recognises the complement of the set of universally-accepting words for $\A$---it is enough to guess a path that is not accepting, and have the acceptance condition that is the dual of $\A$'s condition. 
Formally, for an alternating Rabin (resp.\ parity) automaton $\A=\tuple{\Sigma, Q, \iota, \delta, \alpha}$, we define the nondeterministic Street (resp.\ parity) automaton $\S=\tuple{\boxes_\A, Q, \iota, \delta_\S, \overline{\alpha}}$, where $\overline{\alpha}$ is the dual of $\alpha$ and $\delta_\S$ is defined as follows.
For every states $q,q'\in Q$ and box $\beta\in\boxes_\A$, we have $q'\in\delta_\S(q,\beta)$ iff $\tuple{q,q'}\in\beta$.

Now, one can translate $\S$ to an equivalent deterministic parity automaton $\B'$ with $2^{O(nk \log nk)}$ states \cite{Pit07} (resp.\ $2^{O(n \log n)}$ states~\cite{CZ12,SV14}), and then complement the acceptance condition of $\B'$, getting the required automaton $\B$.

Since nondeterministic coB\"uchi automata can be determinised into deterministic coB\"uchi automata, if $\A$ is B\"uchi, so is $\B$. In general, the parity index of the automaton $\D$ is linear in the number of transitions~of~$\A$.
\end{proof}

We now proceed to the construction of the automaton $\boxA$ of \cref{thm:exp-gfg-dealt}. It is the same as the automaton $\B$ of \cref{lem:prop-of-d}, except that the alphabet is $\Sigma$ and the transition function is defined as follows: For every state $p$ of $\boxA$ and $\letter\in\Sigma$, we have $\delta_{\boxA}(p,\letter):= \cup_{\beta\in\boxes_{\tuple{\A,\letter}}} \delta_{\B}(p,\beta)$.

In other words, the automaton $\boxA$ reads a~letter $\letter$, nondeterministically guesses a~box $\beta\in\boxes_{\A,\letter}$, and follows the transition of $\B$ over $\beta$. Thus, the runs of $\boxA$ over a~word $w\in\Sigma^\omega$ are in bijection between sequences of boxes $(\beta_i)_{i\in\Nat}$ such that $\beta_i\in\boxes_{\A,w_i}$ for $i\in\N$.


Fix an infinite word $w\in\Sigma^\omega$. Our aim is to prove that $w\in L(\A)\Leftrightarrow w\in L(\boxA)$. 

\lempositionalstrategiesandruns*

\begin{proof}
Consider a run of $\boxA$ over $w$, and observe that it corresponds to a sequence of boxes $\beta_0,\ldots$. Notice that each box $\beta_i$ corresponds to Eve's choices in $\A$ over $w_i$, and therefore provides a positional strategy for Eve in the one-step arena $R_{w_i}\times \A$. The sequence of these choices provides a positional strategy for Eve in $R_w\times \A$. 

Dually, given a positional strategy for Eve in $R_w\times \A$, one can extract a sequence of strategies for Eve in the one-step arenas $R_{w_i}\times \A$, and each of them corresponds to a box $\beta_i$. \cref{prop:choice-to-strat} shows that each path in $\beta_0,\ldots$ corresponds to a play consistent with the constructed strategy and vice versa: each play gives rise to a path.

Now, a run is accepting if and only if the sequence of boxes is universally accepting, which means exactly that all the plays consistent with the corresponding strategy are winning.
\end{proof}

We now show that the automaton $\boxA$ is also GFG.

\NotNeeded{
\begin{remark}
The above alternation-removal procedure does not work for an alternating Streett automaton $\A$, even if ignoring the issue of GFGness: Since Streett games are not positionally determined for Eve, the acceptance game of $\A$ over a word $w$ is not positionally determined for Eve.
That is, a winning strategy for Eve in this game might make different choices at a position $\tuple{q,w_i}$, depending on the path leading to this position. Such a strategy does not correspond to a sequence of boxes, and therefore the analogous of \cref{lem:PositionalStrategiesAndRuns} does not hold.
\end{remark}
}

\lemgfgpreservation*

\begin{proof}
Let $\stratE$ be a positional winning strategy for Eve in her expanded letter game for $\A$ (over the arena $R^\ast_{A,\Sigma}\times \A$). The proof is based on the construction of the function $\stratE'\colon \Sigma^+\rightarrow \boxes_\A$, see the paragraph before \cref{def:univ-acc-box}.

Consider the following way of resolving the nondeterminism of $\boxA$: after reading $w\in\Sigma^\ast$, when the next letter $\letter\in\Sigma$ is provided, the automaton moves to the state $\delta_\B(p,\beta_{w\letter})$ where $\beta_{w\letter}=\stratE'(w\letter)$. Consider an infinite word $w\in L(\A)$ and let $\beta_0,\ldots$ be the sequence of boxes used to construct the run of $\boxA$ over $w$. Lemma~\ref{lem:strat-to-win-boxes} implies that this sequence is universally accepting and therefore, the constructed run of $\B$ must also be accepting.
\end{proof}

\subsection{Single-Exponential Determinisation of Alternating Parity GFG Automata}
\label{ap:ssec:exp-det-of-alt}

The aim of this section is to prove the following determinisation theorem.

\thmdet*

First consider the synchronised product $R_{A,\Sigma}\times \boxAco$, which is a~game with labels of the form $\Sigma\times \Gamma_{\boxAco}$, where $\Gamma_{\boxAco}$ is the parity condition of $\boxAco$. Now, we can treat the automaton $\A$ as an~automaton over the alphabet $\Sigma\times \Gamma_{\boxAco}$ that just ignores the second component of the given letter. Thus, we can define a~game $G(\A)= \big(R_{A,\Sigma}\times \boxAco\big)\times \A$.

Notice that $G(\A)$ is naturally divided into rounds, between two consecutive positions of the form $(v,p,q)$, where $v$ is the unique position of $R_{A,\Sigma}$, $p$ is a~state of $\boxAco$ and $q$ is a~state of $\A$. Such a~round, starting in $(v,p,q)$ consists of first Adam choosing a~letter $\letter$; then Eve resolving nondeterminism of $\boxAco$ from $p$ over $\letter$; and then both players playing the game corresponding to the transition condition $\delta(q,\letter)$ of $\A$.

Let the winning condition of $G(\A)$ say that either the sequence of transitions of $\boxAco$ is accepting or the sequence of transitions of $\A$ is accepting. Since $\A$ is $\EGFG$ and $\boxAco$ is GFG, we know that Eve has a~winning strategy in $G(\A)$: she just plays her GFG strategies in both automata and is guaranteed to win whether the word produced by Adam is in $L(\A)$ or $L(\boxAco)$.

As the winning condition of $G(\A)$ is a disjunction of two Rabin conditions, Eve has a positional winning strategy. Fix such a strategy $\stratE$.

\NotNeeded{
\begin{remark}
There is some magic here, as both the GFG strategies of Eve in $\A$ and in $\boxAco$ may require exponential memory. Yet, when she needs to satisfy the disjunction of the two conditions, no more memory is needed. In a sense, the states of $\A$ provide the memory for $\boxAco$ and the states of $\boxAco$ provide the memory for $\A$.  (cf.\  
\cite[Theorem 4]{BKKS13}).
\end{remark}

The second auxiliary game $G'(A)$ is similar, except that Adam is given control of $\boxA$ and Eve is in charge of letters. That is, in $G'(A)$, the configurations are again pairs $(p,q)$ of states from $\boxA$ and $\A$ respectively, and the game proceeds, starting from the initial states, as follows:
\begin{itemize}
\item Eve chooses a letter $\letter\in \Sigma$;
\item Adam chooses a transition $\trans{p}{\letter}{p'}$ in $\boxA$;
\item Adam and Eve play on the one-step arena over $\letter$ from $q$ to a new state $q'$.
\end{itemize}

This time Adam wins a play consisting of a run $\rho$ of $\boxA$ and a path $\rho'$ in $\A$ if either the path is rejecting of $\A$ is rejecting  or the run of $\boxA$ is accepting.

Accordingly, if $\A$ is GFG, then he can win by using the GFG strategy in $\boxA$ and the $\AGFG$ strategy in $\A$. Then if $w\in L(\A)$, the run in $\boxA$ is accepting, and otherwise the path of $\A$ is rejecting.
As before, he also has a positional winning strategy in $G'$.

}

Now do the same with $\A$ and $\boxA$ for Adam: define $G'(\A)$ as $\big(R_{E,\Sigma}\times \overline{\boxA}\big)\times \A$, where $\overline{\boxA}$ is the automaton $\boxA$ where the transitions are turned from nondeterministic to universal, i.e,\ we replace $\lor$ with $\land$.

Again, in a~round of $G'(\A)$ from a~position $(v,p,q)$: Eve plays a letter $\letter$; Adam resolves nondeterminism of $\boxA$ (i.e., the universality in its dual); then they both resolve the choices in $\A$. Let Adam win $G$ if either the play of $\A$ is rejecting or the run of $\boxA$ is accepting. Again we can ensure that Adam has a~winning strategy in $G'(\A)$, because both automata are GFG: he uses the GFG strategy of $\boxA$ and the $\AGFG$ strategy over $\A$. If the word given by Eve belongs to $L(\A)$ then Adam wins by producing an accepting run of $\boxA$, otherwise he wins by refuting an accepting run of $\A$. Let $\stratA$ be his positional winning strategy in that game.

We are now ready to build the deterministic automaton from a GFG APW $\A$, using positional winning strategies $\stratE$ and $\stratA$ for Eve and Adam in $G'(\A)$ and $G(\A)$, respectively.
 
Let $\D$ be the automaton with states of the form $(q,p_1,p_2)$, with $q$ a~state of $\A$, $p$ a~state of $\boxA$ and $p'$ a state of $\boxAco$. 
A transition of $D$ over $\letter$ moves to $(q',p_1',p_2')$ such that $((q,p_1),(q',p_1))$ is consistent with $\stratA$ and $((q,p_2'),(q',p_2'))$ is consistent with $\stratE$. In other words, 
when reading a~letter $\letter$ in such a~state, the following computations are performed:
\begin{enumerate}
\item We simulate the choices made by $\stratE$ in $G'(\A)$ upon obtaining $\letter$ from Adam. This way we know how to resolve nondeterminism of $\boxAco$ and what to do with disjunctions inside~$\A$.
\item We simulate the choices made by $\stratA$ in $G(\A)$ upon obtaining $\letter$ from Eve. This way we know how to resolve nondeterminism in $\boxA$ and what to do with conjunctions of $\A$.
\item In the end we proceed to a new state of $\A$ and resolved nondeterminism of both $\boxA$ and $\boxAco$.
\end{enumerate}

The acceptance condition of $\D$ is inherited from $\A$.

\lemdeteq*

\begin{proof}
Take a word $w\in\Sigma^\omega$. First assume that $w\in L(\A)$. Eve cannot win a play of the game $G$ with the letters played in $R_{A,\Sigma}$ coming from $w$ using by the first disjunct of her winning condition, since $L(\boxAco)=L(\bar\A)$. Thus, all the plays over $w$ consistent with her winning strategy $\stratE$ in $\G'$ must guarantee that the constructed path of $\A$ is accepting. Thus, the run of the automaton $\D$ over $w$ is accepting.

Now assume that $w\notin L(\A)$. Dually, no play of the game $G'$ with the letters coming from $w$ can produce an accepting run of $\boxA$ over $w$. Thus, the strategy $\stratA$ guarantees that the sequence of visited states of $\A$ is rejecting. Thus, the run of $\D$ over $w$ must be rejecting.
\end{proof}


%% file: Fig_4_expanded.tex
\begin{tikzpicture}
\draw (-1.5, 1) edge[obrace] node{$R_{A,\Sigma}$} (1.5, 1);
\draw (3, 1) edge[obrace] node{$R^\ast_{A,\Sigma}$} (12, 1);

\node[adam, minimum size = 25pt] (v) at (0, -1) {$v$};

\draw[trans] (v) edge[out= 60,in=120, looseness=8] node[above] {$a$} (v);
\draw[trans] (v) edge[out=-30,in= 30, looseness=8] node[right] {$b$} (v);

\node[adam, minimum size = 20pt] (v0) at (7.5, +0.5) {$\epsilon$};
\node[adam, minimum size = 20pt] (v0a) at (6.5, -0.5) {$a$};
\node[adam, minimum size = 20pt] (v0b) at (8.5, -0.5) {$b$};

\node[adam, minimum size = 20pt, scale=0.8] (v0aa) at (6.0, -1.5) {$aa$};
\node[adam, minimum size = 20pt, scale=0.8] (v0ab) at (7.0, -1.5) {$ab$};

\node[adam, minimum size = 20pt, scale=0.8] (v0ba) at (8.0, -1.5) {$ba$};
\node[adam, minimum size = 20pt, scale=0.8] (v0bb) at (9.0, -1.5) {$bb$};

\draw[trans] (v0) edge node[above left]  {$a$} (v0a);
\draw[trans] (v0) edge node[above right] {$b$} (v0b);

\draw[trans] (v0a) edge node[scale=0.8, yshift=2pt, left]  {$a$} (v0aa);
\draw[trans] (v0a) edge node[scale=0.8, yshift=2pt, right] {$b$} (v0ab);

\draw[trans] (v0b) edge node[scale=0.8, yshift=2pt, left]  {$a$} (v0ba);
\draw[trans] (v0b) edge node[scale=0.8, yshift=2pt, right] {$b$} (v0bb);

\newcommand{\dotez}[1]{
\node at (#1, -2.0) {$\vdots$};
}

\dotez{5.5}
\dotez{6.5}
\dotez{7.5}
\dotez{8.5}
\dotez{9.5}
\end{tikzpicture}

%% file: Fig_5_history.tex
\begin{tikzpicture}

   \node (b) at (0, 0.75) {};
	\node[state, scale=0.9] (q0) at (0, 0) {$q_0$};
	\node[state, scale=0.9] (q1) at ($(q0)+(-1,-1)$) {$q_1$};
	\node[state, scale=0.9] (q2) at ($(q0)+(+1,-1)$) {$q_2$};
	\node[state, scale=0.9] (q3) at ($(q0)+(0,-1.75)$) {$q_3$};
	\node[state, scale=0.9] (q4) at ($(q3)+(-1,-1)$) {$q_4$};
    \node[state, scale=0.9] (q5) at ($(q3)+(+1,-1)$) {$q_5$};

\node [bool] (b0) at ($(q0)+(0,-0.5)$) {$\land$};
\node [bool] (b3) at ($(q3)+(0,-0.5)$) {$\lor$};

\draw[trans] (b) edge (q0);
\draw[trans] (b0) edge node[scale=0.8, below] {$t_{1}$} (q1);
\draw[trans] (b0) edge node[scale=0.8, below] {$t_{2}$} (q2);
\draw[trans] (q1) edge  (q3);
\draw[trans] (q2) edge  (q3);
\draw[trans] (b3) edge node[scale=0.8, below] {$t_4$} (q4);
\draw[trans] (b3) edge node[scale=0.8, below] {$t_5$} (q5);
\draw[trans] (q4) edge[bend left=70] (q0);
\draw[trans] (q5) edge[bend right=70] (q0);

\node (AC) at (6,-0.5)  [align=left]{Aceptance condition:\\(Finitely often $t_1$ or Infinitely often $t_4$) and\\(Finitely often $t_2$ or Infinitely often $t_5$)};

\end{tikzpicture}

%% file: AP_5_deciding.tex
\section{Appendix of \cref{sec:deciding}}
\label{ap:deciding}

\subsection{Proof of \cref{thm:exp-time-gfg-alt}}
\label{app:exptime-gfg-alt}

Our aim is to provide an \exptime{} algorithm for deciding if a given alternating parity automaton is GFG.

Recall the construction of the two nondeterministic parity automata $\boxA$ and $\boxAco$ for $L(\A)$ and $L(\A)^c$ respectively, as defined in \cref{ssec:exp-det-of-alt}. We will use these automata to design a~game characterising the fact that $\A$ is both $\EGFG$ and $\AGFG$, i.e,\ $\A$ is just GFG.

Recall that the automata $\boxA$ and $\boxAco$ have exponential number of states in the number of states of $\A$. However, due to \cref{lem:prop-of-d} their parity index is linear in the number of transitions of $\A$. Consider the game $G''=\big(R_{A,\Sigma}\times \boxA\big)\times \boxAco$, i.e,\ the game where Adam plays a~letter and Eve replies with two boxes, one of $\A$ and the other of $\overline{\A}$. Let the winning condition of that game for Eve say that either of the runs of $\boxA$ or $\boxAco$ must be accepting.

\begin{lemma}
Eve has a~winning strategy in $G''$ if and only if $\A$ is GFG.
\end{lemma}

\begin{proof}
Clearly if $\A$ is GFG then both $\boxA$ and $\boxAco$ are GFG as nondeterministic automata. Therefore, one can use strategies witnessing their GFGness to construct a single strategy for Eve in $G''$. This strategy must be winning, because each word proposed by Adam either belongs to $L(\boxA)=L(\A)$ or to $L(\boxAco)=L(\A)^c$.

Now assume that Eve has a~winning strategy in $G''$. This strategy consists of two components: one is a strategy in $R_{A,\Sigma}\times \boxA$ and the other in $R_{A,\Sigma}\times \boxAco$. By the fact that the languages of $\boxA$ and $\boxAco$ are disjoint, the above components are in fact winning strategies in the letter games for $\boxA$ and $\boxAco$ respectively. Thus, by \cref{lem:GFG-B-to-A} we know that $\A$ is both $\EGFG$ and $\AGFG$.
\end{proof}

What remains is to show how to solve the game $G''$ in \exptime{}. Let $n$ be the size of the automaton $\A$. Our aim is to turn it into a~parity game of size exponential in $n$ but with a~number of priorities polynomial in $n$. Then, by invoking for instance \cite{CJKLS17}, we know that such a~game can be solved in \exptime.

\begin{lemma}
\label{lem:aut-for-disjunction}
Let $\Gamma=\{0,\ldots,N\}$ be a~set of priorities. Then, there exists a~deterministic parity automaton of size exponential in $N$, with a~number of priorities polynomial in $N$ that recognises the language $L$ of words $w\in(\Gamma\times\Gamma)^\omega$ that satisfy the parity condition on at least one coordinate.
\end{lemma}

\begin{proof}
It is a~rather standard construction. One possibility is to design a~nondeterministic B\"uchi automaton for $L$ with $N^2$ states. Then, the standard determinisation procedure \cite{Pit07} applied to this automaton gives a~deterministic parity automaton as in the statement.
\end{proof}

Therefore, we conclude the proof of \cref{thm:exp-time-gfg-alt} by taking a~product of the game $G''$ with the automaton from \cref{lem:aut-for-disjunction} and then solving the resulting parity game.

\subsection{Proof of \cref{pro:decide-G2-polynomial}}
\label{app:G2-polynomial}

\proDecideGPoly*

\begin{proof}
We start by constructing a deterministic parity automaton $\D$ of a fixed size that recognises whether a word over the alphabet $\{p_1, p_2, \ldots, p_d\}^3$, describing the priorities of the three paths $\rho_E$, $\rho_A$, $\rho_A'$, satisfies the condition ``either $\rho_E$ satisfies the parity condition or neither $\rho_A$ nor $\rho_A'$ satisfies the parity condition.'' 

For constructing $\D$, first take NBWs $\B_1$ and $\B_2$ that recognise that $\rho_A$ and $\rho_A'$ do not satisfy the parity condition, respectively. Both $\B_1$ and $\B_2$ are of size $O(d)$ and proceed by first waiting in an initial state with rejecting transitions until they guess the maximal odd priority to be seen infinitely often and when no higher priorities are seen; then their guess is rewarded with a B\"uchi transition while higher priorities lead to a rejecting sink.

Then construct an NBW $\B_3$ of size $O(d^2)$ that recognises ``neither $\rho_A$ nor $\rho_A'$ satisfies the parity condition'', as the B\"uchi intersection of $\B_1$ and $\B_2$.

Afterwards, construct an NBW $\B_4$  of size $O(d^2)$ for ``either $\rho_E$ satisfies the parity condition or neither $\rho_A$ nor $\rho_A'$ satisfies the parity condition'', as the disjunction of $\B_3$ and a nondeterministic B\"uchi automaton that recognises whether $\rho_E$ satisfies parity. 

Eventually, determinise $\B_4$ to get the NPW $\D$ of size $O(2^{d^2\log d})$ and $O(d^2)$ priorities \cite{Pit07}.

Solving $\Gt(\A)$ then reduces to solving the parity game $\G$ that results from the product between the arena of $\Gt(\A)$ and the automaton $\D$. 
Notice that $\G$ is of size $O(n^3 |\Sigma| 2^{d^2\log d})$ with $O(d^2)$ priorities. 
Jurdzi\'nski and Lazi\'c's quasi-polynomial algorithm for solving parity games operates in time $O(km^{1+o(1)})$, for $m$, $k$ the size of the game and the number of priorities respectively, when $k$ is in $o(\log m)$. We are in this case, so the overall time complexity of solving $\Gt(\A)$ is in $O(d^2(n^{3}|\Sigma| 2^{d^2\log d})^{1+o(1)})$.
\end{proof}

\subsection{Proof of \cref{lem:GFGtoG2}}
\label{app:GFG-to-G2}

\lemGfgToG*

\begin{proof}
Assume $\A$ is GFG, with strategies $\stratE$ and $\stratA$ witnessing respectively that the nondeterminism and universality of $\A$ are GFG. Eve's strategy $\stratE'$ is to play with her token as if she was playing her letter game with strategy $\stratE$ and to play with Adam's tokens as Adam would play in two disjoint copies of his letter game with strategy $\stratA$.
In other words, Eve resolves the nondeterminism for her token using the strategy witnessing that the nondeterminism in $\A$ is GFG and she resolves the universality for Adam's token according to the strategy witnessing that the universality of $\A$ is GFG.

We claim that this strategy is winning. Indeed, in a play $(\rho_E,\rho_A,\rho_A')$ that agrees with $\stratE'$, if the word is in $L(\A)$, then $\stratE$ guarantees $\rho_E$ is accepting while if the word is not accepting, then $\stratA$ guarantees that both $\rho_A$ and $\rho_A'$ are rejecting. 

Furthermore, if $\A$ if GFG, then so in $\bar \A$, and therefore Eve also wins $\Gt(\bar \A)$.
\end{proof}

\subsection{Proof of Proposition~\ref{pro:alt-to-nd}}
\label{app:G2-to-G2box}

\proAltToNd*
The rest of this section is devoted to a proof of this proposition. The proof relies on an additional intermediate game game $\Gpos(\A)$ in which Eve can win positionally.

First, let $\GM$ be a~deterministic parity automaton over the language $(\boxes_\A)^2$ that recognises the sequences of pairs of boxes where at least one sequence is universally accepting. The automaton $\GM$ allows us to turn the condition ``at least one of the sequences of boxes produced by Adam is universally accepting'' into a parity condition. Let $\delta_\GM\colon Q_\GM\times (\boxes_\A)^2\to Q_\GM$ be the transition function of $\GM$.

Now, the game $\Gpos(\A)$ is very similar to the game $\Gt(\boxA)$ except two differences. First, instead of the first copy of $\boxA$ controlled by Eve, we plug a copy of the automaton $\A$, where Eve controls nondeterminism and Adam controls universality. Second, we use~$\GM$ instead of the respective part of the winning condition of $\Gt(\boxA)$. Let the set of configurations of $\Gpos(\A)$ consist of $Q_\A\times Q_\GM$ and the initial configuration be $(\iota_\A,\iota_\GM)$.

In a turn starting in a configuration $(q,p)\in Q_\A\times Q_\GM$, the following choices are done:
\begin{itemize}
\item Adam chooses a letter $\letter\in \Sigma$;
\item Eve and Adam resolve the whole transition of $\A$ from $q$ reaching a~state $q'$;
\item Adam chooses two boxes $b_1$ and $b_2$ over $\letter$.
\end{itemize}
After such a turn, the new configuration is $(q',\delta_\GM\big(p,(b_1,b_2)\big)$.

%

%

A play of the above game provides is a pair of paths $(\rho_E,\rho_A)$ in $\A$ and $\GM$ respectively, and Eve wins if either $\rho_E$ is accepting or $\rho_A$ is rejecting.
Since the winning condition for Eve is a disjunction of two parity conditions, i.e., a Rabin condition, from~\Cref{prop:RabinPositionalDeterminacy} we obtain the following claim.

\begin{claim}
\label{cl:Gpos-positional}
If Eve wins $\Gpos(\A)$ then she has a positional winning strategy.
\end{claim}

The following two lemmata show how $\Gpos(\A)$ is related to both $\Gt(\A)$ and $\Gt(\boxA)$.

\begin{lemma}
\label{lem:gt-to-gpos}
If Eve wins $\Gt(A)$ then she also wins $\Gpos(\A)$.
\end{lemma}

\begin{proof}
Assume $\stratE$ is a winning strategy for Eve in $\Gt(\A)$. During a play of $\Gt(\A)$ the players construct three paths $(\rho_E,\rho_A,\rho_A')$ of the automaton $\A$. We call them the \emph{paths} of that play. Similarly, during a~play of $\Gpos(\A)$ the players construct three sequences $(\rho,\pi,\pi')$, where $\rho$ is a~path of $\A$, while $\pi$ and $\pi'$ are two sequences of boxes of $\A$. We will say that a play of $\Gt(\A)$ with paths $(\rho_E,\rho_A,\rho_A')$ is \emph{consistent} with a~play of $\Gpos(\A)$ with sequences $(\rho,\pi,\pi')$ if $\rho=\rho_E$ and $\rho$ is a~path in $\pi$, and $\rho'$ is a~path in $\pi'$, see \cref{def:path-in-boxes}.

We can now define Eve's strategy $\stratE'$ in $\Gpos(\A)$ as follows. During a play of $\Gpos(\A)$ with sequences $(\rho,\pi,\pi')$, Eve simulates a~play of $\Gt(\A)$ with paths $(\rho,\rho_A,\rho_A')$ that are consistent with $(\rho,\pi,\pi')$. We will now show how Eve can preserve this invariant. Consider a turn of $\Gpos(\A)$ starting in $(q,p)$ ($q$ is the last state of the path $\rho$) and assume that the simulated play of $\Gt(\A)$ ended in a~configuration $(q,q_1,q_2)$.

The turn of $\Gpos(\A)$ starts with Adam choosing a letter $\letter\in\Sigma$. Assume that in the simulated play of $\Gt(\A)$ Adam has also chosen $\letter$. Based on that, the strategy $\stratE$ knows how to resolve disjunctions in $\delta_\A(q,\letter)$ against any choices made by Adam. Assume that $\stratE'$ plays in exactly the same was in the copy of $\A$ in $\Gpos(\A)$. This gives a~transition $(q,\letter,q')$ that is taken in both games. Now, in $\Gpos(\A)$ Adam provides two boxes $b_1$ and $b_2$ of $\boxA$ --- both represent positional strategies of Eve in one-step game over $\letter$. Let Eve assume that in the simulated play of $\Gt(\A)$, Adam was resolving nondeterminism of the last two copies of $\A$ in the way given by $b_1$ and $b_2$ respectively. The strategy $\stratE$ gives a way to resolve universality in these copies, which leads to two transitions $(q_1,\letter,q_1')$ and $(q_2,\letter,q_2')$. Extend the paths $(\rho,\rho_A,\rho_A')$ with the transitions $(q,\letter,q')$, $(q_1,\letter,q_1')$, and $(q_2,\letter,q_2')$ defined above. This way, we managed to preserve the invariant. This concludes the definition of the strategy $\stratE'$.

We will now prove that the strategy $\stratE'$ is winning in $\Gpos(\A)$. Consider an~infinite play in $\Gpos(\A)$ that agrees with $\stratE'$ and its sequences are $(\rho,\pi,\pi')$. By the invariant, there must exist a play of $\Gt(\A)$ that is consistent with $\stratE$ and has paths $(\rho,\rho_A,\rho_A')$ that are consistent with $(\rho,\pi,\pi')$. Now assume that in the considered play of $\Gpos(\A)$, the sequence of states of $\GM$ is parity accepting (otherwise Eve wins the play). We need to prove that $\rho$ is accepting in that case. But the construction of $\GM$ guarantees that in that case at least one of the sequences of boxes $\pi$ or $\pi'$ is universally accepting. W.l.o.g. assume that $\pi$ is universally accepting. Since $\rho_A$ is a~path of $\pi$, it implies that $\rho_A$ is accepting in $\A$. But in that case, the winning condition of $\Gt(\A)$ guarantees that $\rho$ must be accepting in $\A$. We conclude that if Eve wins $\Gt(\A)$, she wins $\Gpos(\A)$.
\end{proof}

Under the assumption of \cref{pro:alt-to-nd} that Eve wins $\Gt(\A)$, \cref{lem:gt-to-gpos}  together with \cref{cl:Gpos-positional} imply that Eve has a~positional winning strategy in $\Gpos(\A)$. The following lemma concludes the proof of \cref{pro:alt-to-nd}.

\begin{lemma}
\label{lem:pos-gpos-to-gt-box}
If Eve has a~positional winning strategy in $\Gpos(\A)$ then Eve wins $\Gt(\boxA)$.
\end{lemma}

\begin{proof}
Let $\stratE$ be a~positional winning strategy of Eve in $\Gpos(\A)$. We assume that it is defined in all the positions of $\Gpos(\A)$, not only those accessible from the initial position. We now construct a winning strategy $\stratE'$ for Eve in $\Gt(\boxA)$. The structure of $\stratE'$ is obtained directly from $\stratE$ by just storing the state of $\GM$ in the memory of $\stratE'$.

More formally, let $\stratE'$ store in its memory a~state $p$ of $\GM$. At each turn starting in a configuration $(q,q_1,q_2)$ of $\Gt(\boxA)$,
\begin{itemize}
\item Adam chooses a letter $\letter$;
\item Eve chooses a box $b$ over $\letter$ such that each $\trans{q}{\letter}{q'}\in b$ is consistent with $\stratE$ from the position $(q,p)$ for each $q\in Q$ --- since $\stratE$ is positional, for each $q\in Q_\A$ it provides a $q$\=/box and their union $b_1$ is a box;
\item Adam chooses boxes $b_1$ and $b_2$ respectively over $a$ for his two tokens;
\item Eve updates her memory state to $\delta_\GM(p,(b_1,b_2))$.
\end{itemize}

Consider a play that agrees with the above defined strategy $\stratE'$ and provides three sequences of boxes $(\pi_E,\pi_A,\pi_A')$. For every path $\rho$ in $\pi_E$, there is a play of $\Gpos(\A)$ that is consistent with $\stratE$ and gives sequences $(\rho,\pi_A,\pi_A')$. If either $\pi_A$ or $\pi_A'$ is universally accepting, also $\rho$ must be accepting. Then $\pi_E$ is also universally accepting. This means that $\stratE'$ is a winning strategy in $\Gt(\boxA)$.
\end{proof}

This concludes the proof of \cref{pro:alt-to-nd}.

%% file: AP_6_cobuchi.tex
\section{Appendix of \cref{sec:G2-coBuchi}}
\label{ap:G2-coBuchi}

In this section we show that a nondeterministic coB\"uchi automaton $\A$ is GFG if and only if Eve wins $G_2(\A)$. This constitutes a step towards \cref{con:G2-to-GFG}.

Let us fix an NCW $\A=(\Sigma,Q,\iota,\delta,\alpha)$ with a transition function $\delta\colon Q\times\Sigma\to 2^Q$ and a set of accepting transitions $\alpha\subseteq Q\times\Sigma\times Q$ .

We use the following standard extensions of the transition function $\delta$:
\begin{itemize}
\item $\overline{\delta}:2^Q\times\Sigma\to2^Q$ is defined by $\overline{\delta}(X,a)=\bigcup_{p\in X} \delta(p,a)$.
\item $\delta^*:\Sigma^*\to 2^Q$ is defined by induction: $\delta^*(\varepsilon)=\{\iota\}$, and $\delta^*(ua)=\overline{\delta}(\delta^*(u),a)$.
\end{itemize}

First, let us recall results on $G_2(\A)$ from \cite{BK18} that are valid for coB\"uchi automata.

\begin{lemma}[\cite{BK18}]\label{lem:GFGtoG2} 
If $\A$ is GFG then Eve wins $G_2(\A)$. 
\end{lemma}

\begin{lemma}[\cite{BK18}]\label{lem:ktokens}
If Eve wins $G_2(\A)$ then for all $k\in\mathbb N$, Eve wins $G_k(\A)$.
\end{lemma}

From now on, we assume that Eve wins $G_2(\A)$, and aim at proving that $\A$ is GFG.

\subsection{Normalisation}

Along the proof, we will change $\A$ to a different automaton $\A'$, while ensuring that if $\A'$ is GFG then so is $\A$. We will then show that $\A'$ is indeed GFG, reaching the required result.

We will use the following operations to change $\A$:
\begin{itemize}
\item \emph{$G_2$-restriction}: We ``clean'' $\A$ to only have states from which Eve wins the $G_2$ game, namely restrict it to the states $Q' =  \{q \St $ Eve wins $G_2(\A^q) \}$, and remove the remaining states and transitions involving them.

\item \emph{Reachability labelling}: 
States occupied by tokens at a given time in a token game on $\A$ are always ``co-reachable'', namely there is a word on which $\A$ can reach all of them. Therefore, augmenting the states of $\A$ with the set of currently reachable states may intuitively help in analysing token games on $\A$. We accordingly define the NCW $\A'=(\Sigma,Q',\iotar,\delta',\alpha')$ that behaves like $\A$, but with an additional component storing the set of states reachable on the prefix read so far. (Cf.\ the ``augmented subset construction'' of \cite{BK12}.)

That is, we define the components of $\A'$ as follows: $Q'=\{(p,X)\in Q\times 2^Q\mid p\in X\}$, $\delta':Q'\times\Sigma\to 2^{Q'}$ defined for every 	
state $(p,X)$ and letter $a\in\Sigma$ by $\delta'((p,X),a)=\{(p',X')\mid p'\in\delta(p,a), X'=\overline{\delta}(X,a)\}$, and the accepting transitions are 	
$\alpha'=\{\trans{(p,X)}{a}{(q,X')}\mid a\in\Sigma, \trans{p}{a}{q}\in\alpha, X'=\overline{\delta}(X,a)\}$.
Notice that the second component evolves deterministically, and that $L(\A')=L(\A)$.
	
\item \emph{Acceptance tuning}: For a coB\"uchi automaton $\A$, we want to normalise $\A$ to only have accepting transitions that can be used infinitely often in an accepting run. We therefore define the accepting transitions of $\A'$ to be the ones that are accepting in $\A$ and are part of a maximal strongly connected component of accepting transitions of $\A$. In other words, all transitions $\trans{p}{a}{q}$ for which there is no path of accepting transitions from $q$ to $p$ in $\A$ are made rejecting. We will call these latter transitions \emph{SCC-changing}.
\end{itemize}

We first show that the $G_2$-restriction allows Eve to win $G_2$ from every state. Notice that it is not trivial from the definition, as her $G_2$ winning strategy in $\A$ might upfront visit states that are then removed in the $G_2$-restriction.

\begin{lemma}
\label{lem:G2RestrictionPreserve}
Consider a nondeterministic coB\"uchi automaton $\A$ and the automaton $\A'$ that is derived from $\A$ by $G_2$-restriction. Then i) if Eve wins $G_2(\A)$ then she wins $G_2(\A'^{q'})$ for every state $q'$ of $\A'$, and ii) if $\A'$ is GFG then so is $\A$.
\end{lemma}

The proof of this lemma relies on the games $G_k(\A)$ as defined in \cref{rem:k-token-nondet}. We will additionally use the explicit shape of Eves strategies in these games, as in the definition of a strategy in memory from \cref{sec:Preliminaries}. For the sake of simplicity, we base that on the representation of the game $G_2(\A)$ (and analogously $G_k(\A)$) as in~\cite[Lemma~11]{BK18}. This means, that Eve's strategy in $G_k(\A)$ can be represented as $\stratE\colon M \times (\Sigma \times Q^{k+1})  \to Q$, where $M$ is the memory used by $\stratE$.

\begin{proof}\
\begin{enumerate}
\item[i)] Consider a state $p\in Q'$. By the definition of $Q'$, Eve wins $G_2(\A)$ from $(p;p,p)$, and by \cref{lem:ktokens} she also wins $G_4(\A)$ from $(p;p,p,p,p)$. Let $\stratE_4$ be her winning strategy that uses some memory structure $M_4$ as above.

We define a strategy $\stratE_2\colon M_4 \times (\Sigma \times Q^3)  \to Q$ of Eve in $G_2(\A)$, in which she plays like $\stratE_4$, assuming that the first two tokens of Adam in the 4-token game follow her token. The memory structure $M_4$ is maintained accordingly.
That is, when Eve's memory is $m$, her token is in state $p$, Adam's tokens are in states $q_1$ and $q_2$, and he chooses the letter $\letter$, we have $\stratE_2(m, \letter, p;  q_1, q_2)=\stratE_4(m,\letter, p; p, p, q_1, q_2)$. Observe that $\stratE_2$ is a winning strategy, since $\stratE_4$ is.

Now, it must be that all states visited by Eve's token in a play consistent with $\stratE_2$ are in $Q'$, as otherwise a play consistent with $\stratE_4$ can reach a position $(p; p, p, q_1, q_2)$, where  $q\in Q\setminus Q'$, from which Adam can win against $\stratE_4$, by diverting to a winning strategy of him in $G_2(\A)$ from $(q;q,q)$ with his first two tokens.

So $\stratE_2$ is actually a winning strategy in $G_2(\A')$ from $(p;p,p)$, as it never goes outside of $Q'$ (and this strategy even works when Adam is allowed to visit $Q$ and not only $Q'$).

\item[ii)]  It is enough to show that $L(\A)=L(\A')$, as then a winning strategy $\stratE\colon\Sigma^*\to Q'$ of Eve in the letter game for $\A'$ is also a winning strategy for her in the letter game for $\A$. (Every word generated by Adam is either not in their language, or followed by an accepting run of Eve in $\A'$, which is also an accepting run for her in $\A$.)
			
Assume toward contradiction that exists a word $w\in L(\A')\setminus L(\A)$.
We give a winning strategy for Adam in $G_3(\A)$, thereby contradicting \cref{lem:ktokens}.
Adam will play the word $w$, and make his first two tokens follow Eve's token, while the third token will follow an accepting run for $w$.
If Eve stays in $Q'$, then by the choice of $w$ she cannot build an accepting run, and Adam will win the play.
So Eve is forced to leave $Q'$ at some point, and the game reaches a position $(p;p,p,q)$ with $p\notin Q'$. By the definition of $Q'$, this means that Adam has a winning strategy $\stratA$ in $G_2(\A)$ from $(p;p,p)$. Adam can therefore stop playing $w$, and win by playing $\stratA$ against Eve with his first two tokens, while doing arbitrary choices with the third token. 
\end{enumerate}
\end{proof}

We continue with showing that reachability labelling and acceptance tuning do not change the $G_2$ winner, and if they produce a GFG automaton then so was the original one.

\begin{lemma}
\label{lem:OperationsPreserve}
Consider a nondeterministic automaton $\A$ and the automaton $\A'$ that is derived from $\A$ by acceptance tuning or reachability labelling. Then i) if Eve wins $G_2(\A^q)$ for every state $q$ of $\A$ then she wins $G_2(\A'^{q'})$ for every state $q'$ of $\A'$, and ii) if $\A'$ is GFG then so is $\A$.
\end{lemma}

\begin{proof}\
\begin{description}	
\item[Acceptance tuning]
Every run $r$ of $\A$ is also a run of $\A'$, and vice versa, and $r$ is accepting in $\A$ iff it is accepting in $\A'$. This is because any accepting run of $\A$ must eventually stay within a SCC of accepting transitions, and therefore avoid SCC-changing transitions. Therefore, the two required properties directly follow.

\item[Reachability labelling]\
\begin{enumerate}
\item[i)] Eve can simply use her winning strategy in $G_2(\A)$ --- the extra component of $\A'$ does not play any role in the acceptance condition, and evolves deterministically.

	
\item[ii)] Consider a winning strategy $\stratE'\colon\Sigma^*\to Q\times 2^Q$ of Eve in the letter game for $\A'$. Then Eve can win the letter game for $\A$, by using the strategy $\stratE\colon\Sigma^*\to Q$ that is derived from $\stratE'$ by ignoring the second component of the image. (Since the second component, consisting of the reachable states, evolves deterministically, Eve can compute it in the memory of her strategy.)	
\end{enumerate}
\end{description}
\end{proof}

\paragraph*{The automaton $\Ar$}

We continue with considering the automaton $\Ar=(\Sigma,\Qr,\iotar,\delr,\alpr)$ that is derived from $\A$ by first performing $G_2$-restriction, then reachability labelling, and finally acceptance tuning. We aim to show that it is GFG, which will show by \cref{lem:G2RestrictionPreserve,lem:OperationsPreserve} that $\A$ is GFG.

Notice that since Eve wins $G_2(\A)$, we have in particular $\iotar\in \Qr$.
In the sequel, we will use $\reach{q}$ to denote a state of $\Ar$ of the form $(q,X)$ with $q\in Q$ and $X\subseteq Q$.
The second component of a state in $\Qr$ (the $X$ in $(q,X)$ ) is deterministically determined (it is the subset construction on the part of $\A$ that was not removed in the normalisation). For a finite word $u$, we shall use $\delrr(u)$ to denote the component $X$ of a state $(q,X)$ reached by $\Ar$ reading $u$.

\subsection{Safety Game and Deterministic Runs}

We analyse the different regions of the automaton $\Ar$ with respect to states being ``safe'' for Eve and states from which she can have some ``partially deterministic'' choices. 

The following constructions and arguments refine the corresponding ones from~\cite{KS15}.

\begin{definition}
Consider an NCW $\C$ with a set $P$ of states. We define the \emph{safety game} $\Gsafe(\C)$ on $\C$ as in~\cite{KS15}: The game is played on $P^2$, and a turn from a configuration $(p,q)$ is played as follows:
\begin{itemize}
\item Adam chooses a letter $\letter\in\Sigma$,
\item Eve chooses a transition $\trans{p}{\letter}{p'}$
\item Adam chooses a transition $\trans{q}\letter{q'}$
\end{itemize}
If the transition $\trans{q}\letter{q'}$ chosen by Adam is rejecting then Eve wins the game immediately.
If $\trans{p}{\letter}{p'}$ is rejecting and $\trans{q}\letter{q'}$ is accepting, Eve loses the game immediately. Otherwise, the game moves to the position $(p',q')$ and a new round starts. Eve wins any infinite play. 
\end{definition}

The above game can again be defined as $(R_{A,\Sigma}\times \C)\times \overline{\C}$ with an appropriate winning condition.

Notice that $\Gsafe(\Ar)$ is a safety game for Eve, and in particular if she wins the game, she can do it with a positional strategy.

We will denote by $\Wsafe\subseteq \Qr^2$ the winning region of Eve in $\Gsafe(\Ar)$.
We show next that for every state in $\Ar$, there is a corresponding safe state sharing the same reachability\=/component. The proof is analogous to the proof of~\cite[Lemma~53 in Appendix~E.5]{KS15}, except that we additionally need to keep track of the component $X$ in the states in $\Qr$.

\begin{lemma}
\label{lem:safewin}
For all $(q,X)\in \Qr$, there exists $p\in Q$ such that $((p,X),(q,X))\in \Wsafe$.
\end{lemma}

\begin{proof}
Assume toward contradiction that there is a state $\reach{q}=(q,X)\in\Qr$ such that for all $p\in Q$, $((p,X),(q,X))\notin \Wsafe$. Since Eve wins $\Gt(\A)$ from each state, she also wins $G_1(\A)$ from $(\reach{q};\reach{q})$. We shall build a winning strategy $\stratA$ for Adam in $G_1(\Ar)$ from $(\reach{q};\reach{q})$, to obtain contradiction.

By the assumption on $\reach{q}$, we have $(\reach{q},\reach{q})\notin\Wsafe$. The strategy $\stratA$ of Adam will start by playing in order to win $\Gsafe$ from $(\reach{q},\reach{q})$.
This means that $\stratA$ guarantees to build a partial play $\trans{(\reach{q},\reach{q})}{u_1}{((p_1,X_1),(q_1,X_1))}$, where a rejecting transition has been seen only on Eve's moves.

Since $\reach{q}$ and $(q_1,X_1)$ are in the same SCC (accepting transitions do not change SCC), 
Adam can now play a word $v_1$ allowing his token to go back to $\reach{q}$.
By the acceptance tuning of $\Ar$, Adam can ensure that the partial run along $v_1$ sees only accepting transitions.

The play therefore reaches a position $((p_1',X), (q,X))$. Notice that the second component $X$ is the same, as it evolves deterministically according to the input word read so far (here $u_1v_1$).
By the assumption, we again have $((p_1',X), (q,X))\notin \Wsafe$. Adam can therefore reiterate the previous strategy: first play in order to win $\Gsafe$ from there, forcing Eve to witness a rejecting transition; then go back to $\reach{q}$ with his token, without seeing any rejecting transition on the loop. This reaches a position $((p_2',X), (q,X))$. Repeating this strategy ad infinitum constitutes the winning strategy $\stratA$, as Eve will be forced to see infinitely many rejecting transitions, while Adam will not see any.
\end{proof}

Let us define $\Ssafe\subseteq Q_r$ by $\Ssafe=\{ (p,X) \St ((p,X),(p,X)) \in\Wsafe\}$. 

From Lemma \ref{lem:safewin}, we deduce the following:
\begin{lemma}\label{lem:safepos}
For all $(q,X)\in Q_r$, there exists $p\in Q$ such that $((p,X),(q,X))\in\Wsafe$ and $(p,X)\in \Ssafe$.
\end{lemma}
\begin{proof}
Let $(q,X)\in Q_r$. By Lemma \ref{lem:safewin}, there exists $p_1\in Q$ such that $((p_1,X),(q,X))\in \Wsafe$. Again, there exists $p_2\in Q$ such that $((p_2,X),(p_1,X))\in \Wsafe$.
Iterating this construction builds a sequence $p_1,p_2,p_3,\dots$. Since $Q$ is finite, there exists $i<j$ such that $p_i=p_j$.
As it is shown in \cite{KS15} that $\Wsafe$ is transitive, we obtain that $p_i\in \Ssafe$ and $((p_i,X),(q,X))\in \Wsafe$. Therefore, taking $p=p_i$ concludes the proof.
\end{proof}
Notice that due to \cref{lem:safepos}, we have in particular that the initial state of $\Ar$, namely $(\iota, \{\iota\})$, is in $\Ssafe$, since $\iota$ is the only state of $\A$ that belongs to a state of $\Ar$ in which the second component is $\{\iota\}$.

We continue with another refinement of a result from \cite{KS15}:
\begin{lemma}\label{lem:safedet}
There exists a partial deterministic transition function $\deltadet\colon\Ssafe\times\Sigma\to \Ssafe$, where $\deltadet\subseteq \delr$, such that for all $w\in L(\Ar)$, there is a decomposition $w=uv$ and a state $s\in Q$, such that $(s,\delrr(u))\in \Ssafe$ and $\deltadet$ accepts $v$ from $s$ without any rejecting transition.
\end{lemma}

\begin{proof}
The function $\deltadet$ is defined by $\deltadet(\reach{p},a)=\sigmasafe(\reach{p},\reach{p},a)$, where $\sigmasafe$ is a positional winning strategy of Eve in $\Gsafe(\Ar)$.

Let $w=a_1a_2\dots\in L(\Ar)$, and $\rho=(p_0,X_0)(p_1,X_1)(p_2,X_2)\dots$ be an accepting run of $\Ar$ on $w$.
Let $k\in\mathbb N$ such that the last rejecting transition in $\rho$ occurs before position $k$. Let $u=a_1a_2\dots a_k$ and $v=a_{k+1}a_{k+2}\dots$, so $w=uv$.
After reading $u$, the run $\rho$ reaches a state $(p_k,X_k)$ with $X_k=\delrr(u)$.
By Lemma \ref{lem:safepos}, there exists $s_k\in Q$ such that $(s_k,X_k)\in \Ssafe$ and $((s_k,X_k),(p_k,X_k))\in\Wsafe$.

We now build by induction a sequence $(s_i)_{i\geq k}$, describing the run of $\Ar$ yielded by $\deltadet$ on $v$ from $(s_k,X_k)$. We show that this run does not contain rejecting transitions, thereby proving the Lemma.
To do so, we show the following invariant $(P_i)$: the run yielded by $\deltadet$ on $a_{k+1}\dots a_i$ does not contain rejecting transitions, and the remaining suffix $a_{i+1}a_{i+2}\dots$ can be accepted in $\Ar$ without rejecting transition from $(s_i,X_i)$.
For the induction base, we need to show that $(P_k)$ is true. The first part is trivial, and the second part follows from the fact that $((s_k,X_k),(p_k,X_k))\in\Wsafe$. Indeed, if from this position in $\Gsafe(\Ar)$, Adam plays $v$ and follows the suffix of $\rho$ from position $k$, Eve must accept $v$ from $(s_k,X_k)$ without seeing any rejecting transition, witnessing the wanted property.

For the induction step, assume $(P_i)$ holds on $(s_i,X_i)$, and let $(s_{i+1},X_{i+1})=\deltadet((s_i,X_i),\allowbreak a_{i+1})$ (recall that the second component evolves deterministically, so $\deltadet$ only chooses the first component $s_{i+1}$). We first need to show that this transition is well-defined and not rejecting.
From $(P_i)$, we know that $a_{i+1}a_{i+2}\dots$ can be accepted without rejecting transitions from $(s_i,X_i)$.
This means that $\sigmasafe((s_i,X_i),(s_i,X_i),a_{i+1})$ must not be rejecting, as otherwise the strategy $\sigmasafe$ would not be winning, since Adam could play an accepting transition from $(s_i,X_i)$ and immediately win the play. By definition of $\deltadet$, we obtain that the transition $\trans{(s_i,X_i)}{a_{i+1}}{(s_{i+1},X_{i+1})}$ is well-defined and accepting.
It remains to show the second part of $(P_{i+1})$, i.e., $a_{i+2}a_{i+3}\dots$ can be accepted from $(s_{i+1},X_{i+1})$ without any rejecting transition in $\Ar$.
Again, consider the play of $\Gsafe$ starting from $((s_i,X_i),(s_i,X_i))$ with Adam playing $a_{i+1}$ and Eve playing $\sigmasafe((s_i,X_i),(s_i,X_i),a_{i+1}))$ to $(s_{i+1}, X_{i+1})$.
If the suffix  $a_{i+2}a_{i+3}\dots$ cannot be safely accepted from $(s_{i+1},X_{i+1})$, then Adam can just play the accepting safe run from $(s_i,X_i)$ on $a_{i+1}a_{i+2}\dots$ (which exists by $P_i$), and win the game $\Gsafe$ against $\sigmasafe$. This is a contradiction, so $(P_{i+1})$ must hold.

This achieves the proof that $\deltadet$ builds a run without rejecting transition from $(s_i,X_i)$ on $v$.
\end{proof}

\subsection{Adam's Strategy $\stratA$ in the Letter Game}

Let us assume toward contradiction that although Eve wins $G_2(\Ar)$, $\Ar$ is not GFG, so Adam has a finite-memory winning strategy $\stratA$ in the letter game on $\Ar$.
Let $M$ be the memory used in $\stratA$, i.e., $\stratA$ has type $\Qr\times M\to \Sigma$ (together with a memory update function), and $m_0$ be the initial memory state.

We will explicit a property of $\stratA$ linked with the strategy $\deltadet$ from Lemma~\ref{lem:safedet}. 

\begin{definition}\label{def:det-breakpoint}
Consider a word $u=a_1a_2\dots a_{|u|}\in\Sigma^*$. We say that positions $i_1<i_2<\dots<i_n$ are \emph{det-breakpoints} of $u$ if for every $j\in [1,n{-}1]$ and $s\in Q$ for which $(s,\delrr(a_1\dots a_{i_j}))\in \Ssafe$, the run yielded by $\deltadet$ from $(s,\delrr(a_1\dots a_{i_j}))$ on $a_{i_{j}+1}\dots a_{i_{j+1}}$ is not defined or witnesses a rejecting transition.
\end{definition}

\begin{lemma}\label{lem:taubreak}
There is a constant $N$ depending only on $\Ar$ and $M$, such that if $u$ is a finite word produced by the strategy $\stratA$ then any sequence of det-breakpoints of $u$ has length smaller than $N$.
\end{lemma}

\begin{proof}
Let $N=|M|\cdot |Q_r|+2$, and consider a partial run $\iotar \xrightarrow[]{a_1} p_1 \xrightarrow[]{a_2}  p_2\xrightarrow[]{a_3}\dots \xrightarrow[]{a_{l}} p_l$ of $\Ar$ on a finite word $u=a_1,\dots, a_l$ that is generated by $\stratA$. Assume toward contradiction that $u$ has $N$ det-breakpoints $i_1<i_2<\dots<i_N$, and let $m_j$ be the memory state of $\stratA$ at step $i_j$.

By the choice of $N$, there must be $j<t\in[1,N]$, such that $(m_j,p_j)=(m_t,p_t)$.
This means that when reaching $p_t$, Eve can repeat the play from $p_j$ to $p_t$, forcing $\stratA$ to produce the same letters $a_{i_j+1}\dots a_{i_t}$ in a loop, producing an infinite word $w$. Moreover, each occurrence of this loop contains a det-breakpoint, so $w$ contains infinitely many det-breakpoints.
By Lemma \ref{lem:safedet}, it follows that $w\notin L(\Ar)$. This contradicts the fact that $\stratA$ is winning in the letter game, as $\stratA$ must always produce a word from $L(\Ar)$ in order to win.
\end{proof}

\subsection{Limit Strategy $\sigmainf$}

Recall that by Lemma \ref{lem:ktokens}, Eve wins $G_k(\Ar)$ for all $k\in\mathbb N$. Let $\Wk\subseteq \Qr^{k+1}$ be the winning region of Eve in $G_k(\Ar)$, and $\stratE_k$ a winning strategy for Eve in $G_k(\Ar)$ from the initial position $(\iotar;\iotar,\dots,\iotar)$.

\begin{proposition}\label{prop:Wk}
Let $(p;q_1,q_2,\dots,q_k)\in \Wk$, let $i\in[1,k]$, and let $\pi_i:[1,i]\to[1,k]$ be an injective function.
Then $(p;q_{\pi_i(1)},q_{\pi_i(2)},\dots,q_{\pi_i(i)})\in W_i$.
\end{proposition}
\begin{proof}
Straightforward.
\end{proof}

\begin{lemma}\label{lem:sigmainf}
There is a strategy\footnote{In \cref{ap:Preliminaries}, we formally defined a ``strategy'' with respect to a specific game. Here we abuse the term ``strategy'' to refer to a general total function on finite words.} $\sigmainf: \Sigma^*\to \Qr$, such that for all $u\in\Sigma^*$, if $\delr^*(u)=\{q_1,q_2,\dots,q_{|\Qr|}\}$, then we have $(\sigmainf(u);q_1,\dots,q_{|\Qr|})\in W_{|\Qr|}$.
\end{lemma}

Notice that in the above definition the cardinality of $\delr^*(u)$ might be smaller than $|\Qr|$ and then some states $q_1,\ldots,q_{|\Qr|}$ repeat.

\begin{proof}
We build $\sigmainf$ by induction on $u$, starting with $\sigmainf(\varepsilon)=\iotar$.

Let us define a strategy $\taukunif:\Qr^k\times\Sigma\to \Qr^k$ for moving $k$ tokens in $\Ar$, by dispatching them uniformly at each nondeterministic choice (remaining tokens are dispatched arbitrarily, for instance using some fixed order on the states). For instance if a state $p$ contains $10$ tokens, and its possible transitions on a letter $a$ are $\delta(p,a)=\{q_1,q_2,q_4\}$, then $\taukunif$ can send $4$ tokens to $q_1$, $3$ tokens to $q_2$, and $3$ tokens to $q_4$.
Let us define a strategy $\sigmakunif:\Sigma^*\to \Qr$ for Eve in the letter game on $\Ar$, as follows.
The memory $M=\Qr^k$ of $\sigmakunif$ consists of $k$ tokens, updated according to $\taukunif$.
The choices made by $\sigmakunif$ are then simply the choices made by $\sigma_k$ against these $k$ tokens.
More formally, $\sigmakunifM: \Qr\times M\times \Sigma\to \Qr$ is defined by $\sigmakunifM(p,m,a)=\sigma_k(p,m,a)$, and its update fuction is induced by $\taukunif:M\times \Sigma\to M$.
The strategy $\sigmakunif:\Sigma^*\to \Qr$ is then defined from $\sigmakunifM$ in a canonical way, using initial memory state $m_0=(\iotar,\dots,\iotar)$, and initial state $\sigmakunif(\epsilon)=\iotar$.

We will preserve the following invariant while building $\sigmainf$: for every finite word $u$, there is an infinite set $I_u\subseteq\mathbb N$ such that for all $k\in I_u$, $\sigmainf$ yields the same run as $\sigmakunif$ on $u$.
This invariant guarantees the statement of the lemma (using Proposition \ref{prop:Wk}), since as soon as $k$ is big enough, all states from $\delr^*(u)$ are each reached by $|\Qr|$ tokens when playing $\taukunif$, and $\sigmakunif$ must always stay in the winning region $\Wk$.

We start with $I_\varepsilon=\mathbb N$, for which the invariant trivially holds.

Assume it holds for $u$ with some infinite set $I_u$, and let $a\in\Sigma$. Let $p=\sigmainf(u)$. For each $k\in\mathbb N$, let $p_k=\sigmakunif(u)$. There exists $q\in \Qr$ such that for infinitely many $k\in I_u$, we have $p_k=q$. We set $\sigmainf(ua)=q$, and $I_{ua}=\{k\in I_u\mid p_k=q\}$.

This maintains the invariant, and thus we can conclude the proof by induction.
\end{proof}

\subsection{Playing Against $\tau$}
\label{ssec:playing}

We will now describe a strategy $\stratE$ for Eve in the letter game of $\Ar$, so that the play yielded by $\stratE$ playing against $\stratA$ is winning for Eve. This will contradict the assumption that $\stratA$ is a winning strategy, leading to the conclusion that $\Ar$ is GFG.

Let $N$ be the constant from \cref{lem:taubreak}. 
The strategy $\stratE$ will intuitively play the $N$-tokens game in $\A_r$ against $N$ imaginary \emph{main tokens} $q_1$, $q_2$, \ldots $q_N$. Here ``imaginary'' means that these tokens exist only in the memory of Eve, and are not part of the actual game arena. So if $u$ is the finite word read so far and $p$ is the current state of Eve, when $\stratA$ produces a new letter $\letter$, Eve will choose a successor state $p'$ by setting $p'=\stratE_N(u\letter,p,q_1,\dots,q_N)$, where $\stratE_N$ is a winning strategy of Eve in $G_N(\Ar)$. 

If $\stratA$ eventually produces a word $w$ not in $L(\A)$, Eve wins by the definition of the letter game. We may thus consider the case where $w\in L(\A)$. Then, since $\stratE_N$ is a winning strategy for Eve, by the definition of tokens game, either the path produced by $\stratE$ is accepting, which is what we target, or all of the paths generated by the imaginary tokens are rejecting. We shall thus ensure that at least one of the paths generated by a $q_i$ token is accepting.

The $N$ main tokens will aim at producing an accepting path in turns, starting with $q_1$; if $q_1$ seems to fail, moving to $q_2$; and so on; until giving the last chance to $q_N$. At each point of time the index of the \emph{active token} is denoted by $j$, and $q_j$ plays the $|\Qr|$-tokens game on $\Ar$ against yet other $|\Qr|$ imaginary  \emph{deterministic tokens} $d_1$, $d_2$, \ldots $d_{|\Qr|}$. 
That is, the new state of the active token is $q_j'=\stratE_{|\Qr|}(u\letter,q_j, d_1,\dots,d_{|\Qr|})$. Analogously to the previous step, since $\stratE_{|\Qr|}$ is a winning strategy for Eve in the $|\Qr|$-token game on $\Ar$, the path of $q_i$ is guaranteed to accept if at least one path of the deterministic tokens that it plays against is accepting.

Each of the ``awaiting'' main tokens $q_{j+1}, \ldots q_N$ should remain in the ``safe area'' of Eve in the $|\Qr|$-tokens game, namely in $W_{|\Qr|}$, until its turn arrives. This ``waiting in the safe area'' is done according to the strategy $\sigmainf$. That is, for every $i\in[j+1,N]$, we set the new state of the $i$-th main token to $q_i'=\sigmainf(ua)$. 
The ``discarded'' main tokens $q_1,\ldots,q_{j-1}$ proceed arbitrarily to a new state compatible with $\delr$. That is, for every $i\in[1,j-1]$, we set the new state of the $i$-th main token to $q_i'\in\delr(q_i,a)$.  

What did we get so far? Instead of directly producing an accepting run for Eve in her letter game, moving around the token $p$, we aim at producing an accepting path for one of the main tokens $q_1,\ldots,q_N$, which we again reduce to producing an accepting path for one of the deterministic tokens $d_1$, \ldots $d_{|\Qr|}$. What is it good for? The first reduction allows to try out $N$ paths instead of a single one. The second reduction allows not to only consider a connected paths, but also paths with up to $N$ ``jumps'' between co-reachable states. That is, whenever the active token is replaced, namely when the index $j$ of the active token is increased by one, a new token game starts between $q_j$ and the deterministic tokens  $d_1$, \ldots $d_{|\Qr|}$. At this point, each of the deterministic tokens can be changed to a new state not only by following $\delr$ but also by going to any co-reachable state.

The $|\Qr|$ deterministic tokens proceed according to $\deltadet$; if for some token $d_i$, the transition $\deltadet$ is not defined or makes a rejecting transition, the token $d_i$ is no longer ``alive'', and we maintain a set $E\subseteq \Qr$ of the alive deterministic tokens. 
That is, for every $i\in E$, we move the token $d_i$ to $\deltadet(d_i,\letter)$, if this transition of $\deltadet$ is accepting, or to an arbitrary state in $\delr(d_i,\letter)$ otherwise.

If the set $E$ of alive deterministic tokens becomes empty, it is a \emph{breakpoint}, on which $\stratE$ behaves as follows: the active token $q_j$ finishes its turn, and $q_{j+1}$ gets its turn to be the active token and play against the deterministic tokens, which become alive again and are spread across all reachable states in $\Ssafe$. That is, we choose the new states $d_1',d_2',\dots,d_{|\Qr|}'$ of the deterministic tokens in any canonical way such that $\{d_1',d_2',\dots d_{|\Qr|}'\}=\delr^*(ua) \cap \Ssafe$.

To sum up the description of $\stratE$, it uses the infinite memory structure $\Sigma^*\times \Qr^N\times [1,N] \times \Qr^{|\Qr|}\times 2^{[1,|\Qr|]}$, 
where a memory state $m=(u,q_1,\dots, q_N, j,d_1,\dots, d_{|\Qr|},E)$ consists of
\begin{itemize}
\item The prefix $u$ of the word read so far.
\item $N$ main tokens $q_1,\dots,q_N$.
\item The index $j\in[1,N]$ of the currently active main token.
\item $|\Qr|$ deterministic tokens $d_1,\dots,d_{|\Qr|}$.
\item A set $E\subseteq |\Qr|$ of the indexes of the alive deterministic tokens, namely those that encountered only accepting transitions since the last breakpoint.
\end{itemize}

The initial memory state is $m_0=(\varepsilon,\iotar,\dots,\iotar,1,\iotar,\dots,\iotar,\{1\})$, where all tokens are in $\iotar$, the active main token is the first one, and only one deterministic token (of index $1$) is alive.
This memory structure of $\stratE$ is updated as described above, the behaviour of the tokens is illustrated in Fig. \ref{fig:runs}.

\begin{figure}
\centering
\includegraphics[scale=.5]{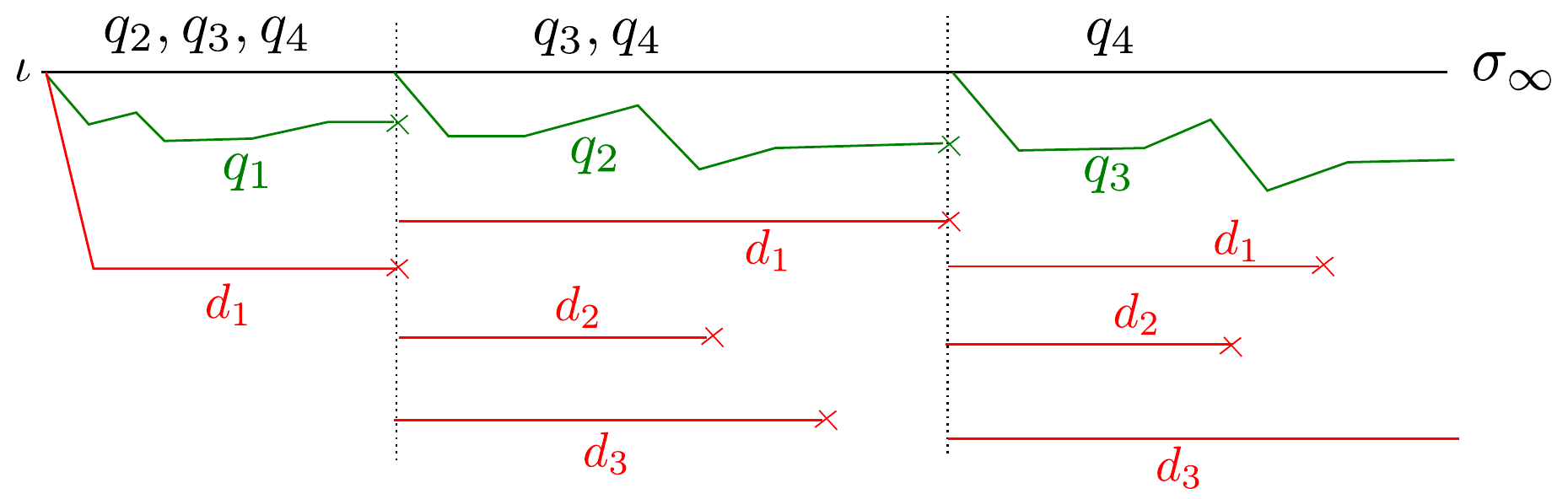}
\caption{An illustration of the behaviour of tokens in the memory structure of the strategy $\stratE$. Awaiting main tokens are represented in black, active tokens in green, and alive deterministic tokens in red. Breakpoints are represented by dashed vertical lines.}
\label{fig:runs}
\end{figure}

We are now ready to prove the main theorem stating the correctness of $\stratE$.
\begin{lemma}\label{lem:sigmatau}
The strategy $\stratE$ wins against $\stratA$ in the letter game of $\A$.
\end{lemma}

\begin{proof}
First of all, notice that by the construction of $\stratE$, if $E$ becomes empty while token $q_i$ is active, then the word $u$ produced by $\stratA$ so far has witnessed $i$ det-breakpoints.
By Lemma \ref{lem:taubreak}, this means that we will never run out of main tokens, and eventually an active token $q_i$ stays active forever.

Let us consider the point of the run where this token $q_i$ becomes active, and let $u$ be the word produced so far. At the breakpoint, deterministic tokens are placed such that $\{d_1,d_2,\dots d_{|\Qr|}\}=
\delr^*(ua) \cap \Ssafe$.
Since $q_i$ followed $\sigmainf$ until now, we have $q_i=\sigmainf(u)$.
By Prop. \ref{prop:Wk} and Lemma \ref{lem:sigmainf}, we have $(q_i,d_1,\dots,d_{|\Qr|})\in W_{|\Qr|}$.
From now on, since no more det-breakpoints occur, some deterministic token $d_h$ will never encounter a rejecting transition, and will safely follow $\deltadet$ forever.
Since $q_i$ now moves according to $\stratE_{|\Qr|}$ against $d_1,\dots,d_{|\Qr|}$, and since $d_h$ follows an accepting run, we obtain that $q_i$ will follow an accepting run.

Finally, since the states assigned to $p$ are chosen according to $\stratE_N$ against $q_1,\dots,q_N$, and since $q_i$ follows an accepting run, we have that $p$ follows an accepting run, and thus the play yielded by $\stratE$ against $\stratA$ is winning for Eve.
\end{proof}

Assuming that Eve wins $G_2(\Ar)$ and $\Ar$ is not GFG leads to a contradiction, we can therefore conclude that $\Ar$ is GFG, and thus $\A$ is GFG by \cref{lem:G2RestrictionPreserve,lem:OperationsPreserve}. We finally have shown \cref{thm:G2cobuchi}.